\documentclass{article}
\usepackage[paper=a4paper,left=25mm,right=25mm,top=25mm,bottom=25mm]{geometry}
\usepackage{a4,amsfonts,latexsym,amscd,amssymb,amsmath, amsthm,bbm,stmaryrd,wasysym,mathbbol}

\usepackage[english]{babel}
\usepackage[utf8]{inputenc}
\UseRawInputEncoding
\usepackage{setspace}
\usepackage{amscd}
\usepackage{bbm}
\usepackage{mathtools}
\usepackage{enumitem}
\usepackage{xcolor}

\allowdisplaybreaks[1]

\newtheorem{Theorem}{Theorem}[section]
\newtheorem{Lemma}[Theorem]{Lemma}
\newtheorem{Corollary}[Theorem]{Corollary}

\newtheorem{Remark}[Theorem]{Remark}
\newtheorem{Proposition}[Theorem]{Proposition}
\newtheorem{Example}[Theorem]{Example}
\theoremstyle{definition}
\newtheorem{Definition}[Theorem]{Definition}

\usepackage{amssymb}

\newcommand\CR{{\mathcal  R}}
\newcommand\R{{\mathbb R}}
\newcommand\X{{\R^d}}

\newcommand\N{{\mathbb N}}
\newcommand\CM{{\mathcal M}}

\newcommand\B{{\mathcal B}}

\newcommand\CK{{\mathcal K}}

\newcommand\La{\Lambda}
\newcommand\la{\lambda}
\newcommand\al{\alpha}

\newcommand\Ga{\Gamma}
\newcommand\ga{\gamma}

\usepackage{hyperref}
\newcommand{\K}{{\mathbb{K}}}
\newcommand{\CF}{\mathcal F}

\newcommand{\CL}{\mathcal L}

\providecommand{\keywords}[1]
{
	\small	
	\textit{Keywords:} #1
}

\title{Classical gases with singular densities}

\author{First Author Name$^{a}$$^{*}$, Second Author Name$^{b}$$^{c}$, etc.$^{a}$$^{c}$ \\
	\small $^{a}$Department, University, City, Country \\
	\small $^{b}$Department, University, City, Country \\
	\small $^{c}$Department, University, City, Country \\\\
	\small $^{*}$Corresponding author: first name, initials, %surname; \tt{email.address}
}
\author{Luca Di Persio \footnote{College of Mathematics, 
		Department of Computer Science, University of Verona, Strada le Grazie 15 - 37134 Verona - Italy; luca.dipersio@univr.it}
	\and  Yuri Kondratiev \footnote{Dragomanov University, Kiev; yukondrat@gmail.com}
	\and Viktorya Vardanyan \footnote{Department of Mathematics, University of Trento, Via Sommarive 14-38123 Povo(TN)-Italy; viktorya.vardanyan@unitn.it}
}

\date{} %leave blank

\begin{document}
	\thispagestyle{empty}
	\maketitle
	
	\begin{abstract}  
	
	We study classical continuous systems with singular distributions of velocities. Radon measures with the infinite mass give these distributions. Positions of particles in such systems are no longer usual configurations in the location space,  leading to the necessity of developing new analytical tools to study considered models.

	\end{abstract} \hspace{10pt}

	\\
	
	
	\section{Introduction}
Let \( \Gamma(\mathbb{R}_0^d \times \X) \) denote the space of configurations for an interacting particle system, where each particle is described by a pair \( (v_x, x) \), with \( v_x \in \mathbb{R}_0^d = \mathbb{R}^d \setminus \{0\} \) representing the velocity and \( x \in \X \) the position. The phase space is then \( \mathbb{R}_0^d \times \X \), and locally finite subsets of this space give the configurations. Specifically, a configuration \( \gamma \in \Gamma(\mathbb{R}_0^d \times \X) \) is a set of pairs \( \gamma = \{(v_x, x)\} \subset \mathbb{R}_0^d \times \X\), where each particle is characterized by its velocity \( v_x \) and position \( x \).

We focus on a subset of configurations, denoted \( \Gamma_p(\mathbb{R}_0^d \times \X) \), called {\it pinpointed configurations}. A configuration is pinpointed if no two particles are positioned at the same location in \(\X \) unless they have the same velocity, i.e., each position \( x \) in the configuration corresponds to at most one velocity \( v_x \).

Further, we define the {\it Plato space}, denoted \( \Pi(\mathbb{R}_0^d \times \X) \), as the subset of \( \Gamma_p(\mathbb{R}_0^d \times \X) \) consisting of configurations where the total velocity within any compact region of \( \X \) is bounded. This constraint ensures that the system's behaviour remains physically plausible by avoiding the occurrence of unbounded velocities within finite spatial regions.

We define the cone of vector-valued discrete Radon measures $\K(\X)$ as a  subset of the space of Radon measures $\CM(\X)$
$$\K(\X)= \bigg\{\eta=\sum_{x\in \tau(\eta)} v_x\delta_x \in \CM(\X) \bigg| x \in \X, \ v_x \in \R^d_0  \bigg\},$$
 where the support of  $\eta$ is defined as
$$ \tau(\eta)= \{x \in \X\;|\; \eta(\{x\})=:v_x(\eta)\}.$$

This paper is dedicated to the in-depth exploration and analysis of the cone $\K(\X)$, which serves as a significant framework for modelling particle systems in real-world scenarios. Defining mathematical structures on $\K(\X)$, however, presents substantial challenges. To address these, we employ the Plato space $\Pi(\R^d_0 \times \X)$, where the elements of this space, termed *ideas*, are mapped through the reflection mapping $\CR$ to corresponding elements in $\K(\X)$, which represent observed objects (i.e., images of ideas). The properties of the object space in $\K(\X)$ are intrinsically tied to the structure of the idea space in the Plato space. This conceptual framework is inspired by Plato's theory of forms, which asserts that the observable world is a projection of a higher realm of ideal forms. The relationship between Plato space $\Pi(\R^d_0 \times \X)$ and the cone of vector-valued discrete Radon measures $\K(\X)$, mediated by the reflection mapping, forms the basis for the analysis of $\K(\X)$.

The nonlinearity and complexity of the infinite-dimensional spaces involved present significant challenges in directly analyzing the dynamics modelled on $\Pi(\R^d_0 \times \X)$. To overcome this, we propose reinterpreting the equations within the finite configuration space $\Pi_0(\R^d_0 \times \X)$, a subset of $\Pi(\R^d_0 \times \X)$. This reformulation simplifies the analysis, and the $K$-transform plays a central role in enabling this transition.

In this paper, we perform harmonic analysis on the space $\Pi(\R^d_0 \times \X)$ and extend this analysis to the cone of vector-valued discrete Radon measures $\K(\X)$ using the reflection mapping $\CR$. Additionally, we investigate the correlation measures and correlation functions associated with probability measures on the cone. This area of study was previously developed by Kondratiev, Lytvynov, and Vershik in \cite{KLV} for the cone of positive discrete Radon measures on a Riemannian manifold, and by Finkelshtein, Kondratiev, Kuchling, Lytvynov, and Oliveira in \cite{PhD}, \cite{ph} for the cone of positive discrete Radon measures on $\X$. Harmonic analysis on configuration spaces over Riemannian manifolds was developed in \cite{MR1914839} by Kondratiev and Kuna, and we apply the corresponding approach to the cone $\K(\X)$.

	\section{Vector Valued Random Discrete Measures}

	\subsection{ From Configuration Spaces  to  Discrete Measures}
	
	Consider the Euclidean space $\X$ and the locally compact pointed 
	space $\R^d_0=\X\setminus \{0\}$.
	Let $\Ga(\R^d_0 \times \X)$ denote the configuration space over 
	$\R^d_0\times \X$, i.e., the set of  all locally finite  subsets 
	$\ga=\{(v,x)\}\subset \R^d_0 \times \X$.
    A subset is considered locally finite if it contains only finitely many points within any compact set. We equip this space with 
	the vague topology as described in \cite{AKR}, and the corresponding $\sigma$-algebra is  denoted by
	$\B(\Ga(\R^d_0 \times \X))$. 
	
	Within all configurations, we focus on a specific subspace $\Ga_{p}(\R^d_0 \times \X)$, referred to as the space of pinpointed configurations. 
	By definition, $\Ga_{p}(\R^d_0 \times \X)$ consists of configurations
	$\ga \in \Ga (\R^d_0 \times \X)$ satisfying the condition that if  $(v_1, x_1), (v_2,x_2)\in \ga$
	and $ (v_1,x_1)=(v_2,x_2)$ with $x_1=x_2$,  then $v_1= v_2$. This means we cannot have two points with the same location in $\X$ but different values in $\R^d_0$. As easily seen
	$\Ga_{p}(\R^d_0 \times \X)\in \B(\Ga (\R^d_0 \times \X))$. For any $\ga\in \Ga_{p}(\R^d_0 \times \X)$,
	we can write the configuration as
	$$
	\ga =\{ (v_x, x)\} ,\; \tau(\ga)= \{x\;|\; (v_x,x)\in \ga\}\subset \X,
	$$
    where $\tau(\ga)$ denotes the projection of the configuration onto its spatial positions in $\X$.
	%It is important to note that $\tau(\gamma)$ is not typically a configuration but rather a discrete set in $\mathcal{X}$, and we will focus on this scenario.
We can also examine smaller subsets of configurations where \(\tau(\gamma)\) remains a configuration. These subsets correspond to {\it marked point processes}, which offer a more tractable framework for study and analysis.
    
    %We can also consider smaller subsets of configurations where $\tau(\gamma)$ is once again a configuration. These configurations are known as marked point processes, which are more straightforward to study and analyze.

	For any configuration  $\ga\in \Ga_p (\R^d_0 \times \X)$ and any compact set $\La\in \B_c(\X)$, we define the 
	local velocity functional as:
	$$
	V_\La (\ga) = \sum_{x\in\tau(\ga)\cap \La} |v_x|\leq \infty.
	$$
   This functional measures the sum of the magnitudes of velocities associated with points in $\ga$ whose spatial positions lie within the compact set $\La$.  Using this functional, we define the Plato space
	$\Pi(\R^d_0\times\X)$ as a subset of the pinpointed configuration space  $ \Ga_p (\R^d_0 \times \X)$. Specifically, 
	
	$$
	\Pi(\R^d_0\times\X) =\{ \ga\in \Ga_p (\R^d_0 \times \X) \; \forall  \La\in \B_c(\X) \;|\;  V_{\La}(\ga)<\infty\}.
	$$
Let $\CM(\X, \X)$ denote the set of all vector-valued Radon measures defined on the Borel subsets $\B(\X)$, such that

	$$
	\CM(\X,\X) \ni \mu: \B(\X) \to \X,\; \forall \La\in \B_c(\X)\;\;|\mu(\La)| <\infty.
	$$
We introduce a mapping from the Plato space $\Pi(\R^d_0 \times \X)$ to the space of vector-valued Radon measures:

	$$
	\Pi(\R^d_0\times\X) \ni \ga=\sum_{(v_x, x)\in \ga} \delta_{(v_x, x)} \to \eta= \CR\ga =\sum_{x\in \tau(\ga)} v_x\delta_x \in \CM(\X),
	$$
	where the image of this mapping is denoted as $\K(\X)$:
	$$
	\K(\X)= \CR(\Pi(\R^d_0\times\X)) \subset \CM(\X).
	$$
The topology of $\K(\X)$ is derived from the topology of the configuration space $\Ga(\R^d_0 \times \X)$.

In this framework, the points of \(\Pi(\R^d_0 \times \X)\) are identified as *ideas*, while their images under the reflection mapping \(\mathcal{R}: \Pi(\R^d_0 \times \X) \to \K(\X)\) are interpreted as *observed objects* (i.e., the mapped representations of the ideas). The structure and properties of the object space \(\K(\X)\) are entirely governed by the underlying Plato space \(\Pi(\R^d_0 \times \X)\), which serves as the foundational framework for all further analyses of \(\K(\X)\).

From a physical perspective, this construction models a system of particles in the space \(\X\), where each spatial position \(x \in \X\) is associated with a velocity (or mark) \(v_x \in \R^d_0\). This representation encapsulates the particle system's spatial configuration and dynamic behaviour.

\subsection{Poisson Measure on $\Ga(\R_0^d\times \X)$ }
	
The classical ideal gas is described by a pair $(v_x, x)$  where $x\in \X$ represents the position and $v_x\in \R_0^d= \X\setminus \{0\}$ represents the velocity of a particle. The phase space of the gas is then given by $\R_0^d\times \X$. Microscopic states of the gas are represented by phase space configurations, which are locally finite subsets of $\R_0^d\times \X$. The space of such configurations is denoted by:
	$$
	\Ga(\R_0^d\times \X)\ni \ga=\{(v_x, x)\}\subset \R_0^d\times \X.
	$$
A macroscopic state of the gas shall be a probability measure on $\Ga(\R_0^d\times \X)$. For the ideal gas, such a measure is a Poisson point measure.
Let us take a Radon measure $\la$ on $\R_0^d$ and a non-atomic  Lebesgue measure $m$ on $\X$. Consider an intensity measure
	$\la\otimes m $ on $\Ga(\R_0^d\times \X)$ and the Poisson measure $\pi_{\la\otimes m}$ on $\Ga(\R_0^d\times \X)$.
	The following Laplace transform may characterize that measure:
	$$
	\int_{\Ga(\R_0^d\times \X)} \exp ( \sum_{(v_x, x)\in \ga}  \psi(v_x, x))\pi_{\la\otimes m}(d\gamma)= \int_{\R_0^d\times \X}
	(\exp \psi(v_x, x) -1) \la(dv)m(dx)
	$$
	for any continuous $\psi$ with bounded support.
The support of this measure is essentially smaller and may be restricted to  the pinpointed configurations:
		$$
	\ga =\{ (v_x, x)\} ,\; \tau(\ga)= \{x\;|\; (v_x, x)\in \ga\} \subset \X.
	$$
	Note that the set $\tau(\ga)$, in general, is a discrete subset of $\X$ but is not
	locally finite, i.e., is not a configuration in $\X$.
	
	To have a physically relevant property of finite local energy,
	 we need to assume an additional property of the measure $\la$.
	Namely, we consider measures with finite first moments
	$$\int_{\R_0^d} |v|\la(dv) <\infty.
	$$
	This situation will be at the centre of our considerations. 
	
	Note that in the case of finite measures $\la$, the Poisson measure
	$\pi_{\la\otimes m}$  may be considered in the framework of marked point 
	processes \cite{Ki}. The latter allows for an easier technical analysis.

	We are interested in defining the following mapping:
	$$
	\ga\mapsto \eta=\sum_{x\in\tau(\ga)} v_x \delta_x
	$$ 

which is a transformation mapping a configuration \(\ga\) to a discrete vector-valued Radon measure on \(\X\). The space of such measures, which becomes random when \(\ga\) is random, exhibits rich mathematical structures involving complex analysis and geometry elements. For the simpler case of real-valued discrete measures, refer to \cite{KLV}. 

The scenario most analogous to the classical case in statistical physics arises with probability measures \(\lambda\), particularly when dealing with Maxwellian (Gaussian) distributions. In these cases, integrating concerning velocities yields a Poisson distribution \(\pi_m\) on the configuration space of positions \(\Ga(\X)\). However, this result is unattainable when the intensity measure satisfies \(\lambda(\X) = \infty\). The distinction between these two cases is fundamental, as the properties of the velocity distribution play a critical role in determining the spatial structure of the model in the case of an infinite measure \(\lambda\)."

    %which transforms a configuration 
%$\ga$ into a discrete vector-valued Radon measure on $\X$. The space of such measures (which are random for random $\ga$) has rich mathematical properties involving complex analysis and geometry.For the more straightforward case of real-valued discrete measures, see \cite{KLV}. The situation in statistical physics that comes closest to the traditional scenario involves probability measures $\la$, mainly when dealing with Maxwell (Gauss) distributions. In these cases, integrating concerning velocities leads to a Poisson distribution $\pi_m$ in the position configuration space $\Ga(\X)$, which is unattainable for an infinite intensity measure $\la(\X)=\infty$. There is a fundamental difference between these two scenarios, as the characteristics of velocity distributions can significantly impact the spatial structure of the model when dealing with an infinite measure $\la$.

\subsection{ Measures on $\K(\X)$}
Probability measures on the space of compact subsets \(\K(\X)\) are fundamental tools for modelling various phenomena in mathematical physics, biology, and the representation theory of infinite-dimensional groups. Our objective is to construct such measures by utilizing the reflection map \(\CR\) applied to carefully selected measures on the Plato space \(\Pi(\R^d_0 \times \X)\). To exemplify this method, we begin with elementary cases related to Poisson measures. Specifically, we first define a non-atomic Radon measure \(\lambda\) on \(\R^d_0\) that satisfies the following conditions:

%Probability measures on the space of compact subsets $\K(\X)$ play a crucial role in modelling various phenomena across mathematical physics, biology, and the representation theory of infinite-dimensional groups. We aim to construct such measures by applying the reflection map $\CR$ to appropriately chosen measures on the Plato space $\Pi(\R^d_0 \times \X)$. To illustrate this approach, we begin by considering simple examples associated with Poisson measures. First, we define a non-atomic Radon measure $\la$ on $\R^d_0$ that satisfies the following properties:
    $$
	\la(\R^d_0) =\infty,
	$$
	$$
	\forall n\in \N\;\;  \int_{\R^d_0} |v|^n \la(dv) <\infty.
	$$
As a concrete example, we take
	$$
	\la(dv)= \frac{1}{|v|^d} e^{-|v|^2} dv,
	$$
which represents a modified version of the Maxwell distribution for velocities in statistical physics, incorporating a singularity at $v = 0$. A broader class of such measures can be expressed in the following form
	$$
	\la(dv) = \frac{1}{|v|^\al} e^{-|v|^\beta} dv,
	$$
where $\alpha \in [d, d+1)$ and $\beta > 0$. These generalizations allow for flexibility in modelling systems with different velocity distributions.

Let $m(dx)$ denote a non-atomic Lebesgue measure on $\X$, and define  the intensity measure:
	$$
	\sigma(dv,dx)= \lambda(dv) m(dx)
	$$
	on $\R^d_0 \times \X$.  This measure combines the velocity distribution $\lambda$ with the spatial measure $m$.
    
 For any $\La \in \B_c(\R^d_0\times \X)$ we have the following disjoint decomposition:
 $$\Ga(\La)=\bigcup_{n=0}^{\infty}\Ga^n(\La),$$
 where $\Ga(\La):=\{\gamma \in \Ga(\R^d_0 \times \X): \gamma \subset \La\}$ is the set of all configurations supported in $\La$ and  $\Ga^n(\La):=\{\gamma \in \Ga(\La): |\gamma|=n \}$ is the set of $n$-point configurations. The Lebesgue-Poisson measure on $\Ga(\La)$ is given by:  
 $$\mathcal{L}_\sigma:= \sum_{n=0}^\infty \frac{1}{n!}\sigma^{(n)},$$
 where $\sigma^{(n)}$ is the symmetric product measure  on $\Ga^n(\La)$, defines as $\sigma^{(n)}:=(\la \otimes m)^{\otimes n} \circ sym_n^{-1}$. Poisson measure on $\Ga(\La)$ is defined as $\pi_\sigma^\La:= e^{-\sigma(\La)}\mathcal{L}_\sigma$.
 
 There is a standard definition of a Poisson measure 
	$\pi_\sigma$ 
	on $\Ga(\R^d_0 \times \X)$, therefore we can proceed as follows: for any $\psi\in C_0(\R^d_0 \times \X)$ (continuous functions with compact supports) define
	
	$$<\psi, \ga>=\sum_{(v,x)\in \ga} \psi(v,x),\;\; \ga \in \Ga(\R^d_0 \times \X).
	$$
	The Poisson measure $\pi_{\sigma} $ is defined via its Laplace transform:
	$$
	\int_{\Ga(\R^d_0 \times \X)} e^{<\psi,\ga>} \pi_{\sigma}  (d\ga)=
	\exp{\int_{\R^d_0 \times \X} (e^{\psi(v,x)} -1) } \la(dv) m(dx).
	$$
	As a consequence of this definition, we have
	$$
	\int_{\Ga(\R^d_0 \times \X)} <\psi,\ga> \pi_\sigma (d\ga) = \int_{\R^d_0 \times \X} \psi(v,x) \la(dv) m(dx)
	$$
	and this relation can be extended to the case of any $\psi\in L^1 (\la\otimes m)$.

	\begin{Lemma}
		The Poisson measure $\pi_{\sigma}$ in concentrated on $\Pi(\R^d_0\times\X)$, i.e.,
		
		$$
		\pi_{\sigma} (\Pi(\R^d_0\times\X))=1.
		$$

	\end{Lemma}
	\begin{proof}
		Take any $\Lambda \in \B_c(\X)$ and define
		$$
		\psi(v,x)= |v| 1_{\La} (x).
		$$ 
		Then, 
		$$\int_{\Ga(\R^d_0 \times \X)} \sum_{(v,x)\in \ga} |v| 1_{\La}(x) \pi_{\sigma} (d\ga)=
		m(\La) \int_{\R_0^d} |v| \la(dv) <\infty.
		$$
		It means that $V_{\La} <\infty$ $\pi_{\sigma}$-a.s.,
		i.e., $\pi_{\sigma} (\Pi(\R^d_0\times\X)) = 1$.

	\end{proof}
To obtain measures on $\K(\X)$, we can employ the {\it pushforward} of measures on $\Pi(R^d_0\times \X)$ through the mapping $\CR$. One crucial step is demonstrating the measurable structures' compatibility on $\Pi(oR^d_0\times \X)$ and $\K(\X)$.
%\subsubsection{Probability Measures on $\K(\X)$}
Having already established the connection between $\K(\X)$ and $\Pi(R^d_0\times \X)$, it is natural to explore further the relationship induced by the mapping 
$\CR$. For instance, one could examine the relationship between the
$\sigma$-algebras $\B(\Pi(R^d_0\times \X))$ and $\B(\K(\X))$.
	
	\begin{Theorem}
		The image $\sigma$-algebra of $\B(\Pi(R^d_0\times \X))$ under $\CR$ and $\B(\K(\X))$ coincide, i.e.,
		\begin{displaymath}
			\B(\K(\X))=\left\{\CR(A\cap\Pi(R^d_0\times \X))\mid A\in\B(\Gamma(R^d_0\times \X))\right\}.
		\end{displaymath}
	\end{Theorem}
	\begin{proof}
		The proof follows directly from the topological considerations presented in \cite{KLV} for the case of real-valued discrete measures.
		
	\end{proof}
	
Denote $\mu_\la$ the image  measure on $\K(\X)$ under  
	the reflection map corresponding to $\pi_\sigma$.
	For $h\in\X$ and $\phi\in C_0(\X) $ introduce a function
	$L_{h,\phi}: \K(\X) \to \R$ via:
	$$
	L_{h,\phi}(\eta)=  <h\otimes \phi, \eta> = \int_{\X} \phi(x) <h,\eta(dx)> = \sum_{x\in\tau(\eta)} \phi(x) <h,v_x>.
	$$
	The definition of the Poisson measure implies
	$$
	\int_{\K} e^{<h\otimes \phi, \eta>} \mu_\la (d\eta)= \exp( \int_{\R^d_0\times \X} (e^{<h,v>\phi(x)} -1) \la(dv) m(dx)).
	$$
Define the function:
	$$
	\Phi_\la^h (r) = \exp({\int_{\R^d_0} (e^{<h,v>r} -1) \la(dv)}), \;\;r\in\R,
	$$
	within the integral can be rewritten as
	$$
	\int_{\K} e^{<h\otimes \phi, \eta>} \mu_\la (d\eta)= e^{\int_{\X} \log (\Phi_\la^h (\phi(x))) m(dx)}.
	$$
The definition of the measure $\mu_\la$
  through its Laplace transform can be 
  understood in terms of the equation above.

\section{Harmonic Analysis on $\Pi(\R^d_0\times \X)$}\label{harmana_pi}

The inherent nonlinearity of infinite-dimensional spaces presents significant challenges in directly analyzing the dynamics formulated on \(\Pi(\R^d_0 \times \X)\). To address this, we propose reformulating the equations within the space of finite configurations, \(\Pi_0(\R^d_0 \times \X)\), which is a proper subset of \(\Pi(\R^d_0 \times \X)\). Applying the \(K\)-transform facilitates this reformulation. A comparable methodology was utilized by Kondratiev and Kuna in \cite{MR1914839} for conducting harmonic analysis on configuration spaces over Riemannian manifolds.
 
\subsection{The $K$-Transform}
 
In what follows, we will introduce auxiliary space connected to the Plato space $\Pi(\R^d_0\times \X)$ via the $K$-transform. Moreover we will show relations between functions on $\Pi(\R^d_0\times \X)$ and $\Pi_0(\R^d_0\times \X)$.
	\begin{Definition}
		The Plato space of finite configurations $\Pi_0(\R^d_0\times \X)$ is defined as:
		\begin{displaymath}
			\Pi_0(\R^d_0\times \X):=\{\gamma\in\Pi(\R^d_0\times \X)\mid |\gamma|<\infty\},
		\end{displaymath}
		where $|\cdot|$ denotes the number of elements in a set. Its topology is induced by the set $\Gamma_0(\R^d_0\times \X)$. 
	\end{Definition}
 The spaces $\Pi(\R^d_0\times \X)$ and $\Pi_0(\R^d_0\times \X)$ play completely different roles. 
 As mentioned before, $\Pi(\R^d_0\times \X)$  represents the space of ideas, whereas $\Pi_0(\R^d_0\times \X)$ is regarded as a mathematical construct that exists alongside $\Pi(\R^d_0\times \X)$.  Moreover, these spaces are also topologically different.\\
We will introduce $n$-point configurations, which are used to decompose the space.
\begin{Definition}

\begin{enumerate}
\item For $n\in\N_0$, the set of $n$-point configurations is defined as:
			\begin{displaymath}
				\Pi_0^{(n)}(\R^d_0\times \X):=\left\{\gamma\in\Pi_0(\R^d_0\times \X)\colon|\gamma|=n\right\}.
			\end{displaymath}
\item For a set $\Lambda\subset\R^d_0\times\X$, the set of all configurations supported in $\Lambda$ is defined as:
			\begin{displaymath}
				\Pi_0(\Lambda):=\left\{\gamma\in\Pi_0(\R^d_0\times \X)\colon\gamma\subset\Lambda\right\}.
			\end{displaymath}
\item A Borel set $A\subset\Pi_0(\R^d_0\times \X)$ is called bounded if there exists $\Lambda\subset\R^d_0\times\X$ compact and $N\in\N$ such that
			\begin{displaymath}
				A\subset\bigcup_{n=0}^N\Pi_0^{(n)}(\Lambda).
			\end{displaymath}
Denote the system of all such sets by $\B_b(\Pi_0(\R^d_0\times \X))$.
\end{enumerate}

We have the following decompositions:
		\begin{displaymath}
			\Pi_0(\R^d_0\times \X)=\bigsqcup_{n=0}^\infty\Pi_0^{(n)}(\R^d_0\times \X)=\bigcup_{\Lambda\in\B_c(\R^d_0\times \X)}\Pi_0(\Lambda),
		\end{displaymath}
where the first union is disjoint and $\B_c(\R^d_0\times \X)$ denotes all Borel subsets of $\R^d_0\times \X$ with compact closure.
\end{Definition}
	To introduce the $K$-transform between $\Pi_0(\R^d_0\times \X)$ and $\Pi(\R^d_0\times \X)$, we first need to define well-defined classes of functions on which this transform can be applied. Additionally, we will introduce a specific class of measures that can extend the $K$-transform to a wider range of functions. We use the notation $B^0(\Pi_0(\R^d_0\times \X))$ to represent the set of all measurable functions $G\colon\Pi_0(\R^d_0\times \X)\to\R$.
	
\begin{Definition}\label{locsupp_pi}
\begin{enumerate}
\item  A function $G\colon\Pi_0(\R^d_0\times \X)\to\R$ is said to be bounded with local support if there exists $C>0$ and $\Lambda\in\B_c(\R^d_0\times \X)$ such that the following estimate holds for all $\eta\in\K_0(\X)$:
			\begin{equation}\label{locbd_pi}
				|G(\gamma)|\leq C\mathbbm{1}_{\Pi_0(\Lambda)}(\gamma).
			\end{equation}
It is important to note that this implies that $G(\gamma)=0$ if $\gamma\cap\Lambda^c\neq\emptyset$. We denote by $B_\mathrm{ls}(\Pi_0(\R^d_0\times \X))$ all measurable functions $G\colon\Pi_0(\R^d_0\times \X)\to\R$ which are bounded with local support.
\item A function $G\colon\Pi_0(\R^d_0\times \X)\to\R$ is called bounded with bounded support, if there exists $\Lambda\in\mathcal{B}_c(\R^d_0\times \X), N\in\N$ and $C>0$ such that
			\begin{equation}\label{bbs}
				|G(\gamma)|\leq C\mathbbm{1}_{\Pi_0(\Lambda)}(\gamma)\mathbbm{1}_{\{|\gamma|\leq N\}}(\gamma), \end{equation}
i.e. $G(\gamma)=0$ whenever $|\gamma|>N$ or $\gamma\cap\Lambda^c\neq\emptyset$. Denote the space of all such measurable functions by $B_\mathrm{bs}(\Pi_0(\R^d_0\times \X))$. Evidently, it follows that we have $B_\mathrm{bs}(\Pi_0(\R^d_0\times \X))\subset B_\mathrm{ls}(\Pi_0(\R^d_0\times \X))$.
\item A measure $\rho$ on $\Pi_0(\R^d_0\times \X)$ is called locally finite if for any $\Lambda\in\B_c(\R^d_0\times \X)$ and for any $m\in\N_0$, the value of $\rho(\Pi_0^{(m)}(\Lambda))$ is finite. Similarly, $\rho(A)$ is finite for all bounded measurable sets $A\subset\Pi_0(\R^d_0\times \X)$. The space of all locally finite measures on $\Pi_0(\R^d_0\times \X)$ is denoted by $\CM_\mathrm{lf}(\Pi_0(\R^d_0\times \X))$.
\end{enumerate}
\end{Definition}
 In the following, we define the $K$-transform and its main properties.

	\begin{Definition}[\cite{MR1914839}]
		\label{K transform}
		Let $G\in B_\mathrm{ls}(\Pi_0(\R^d_0\times \X))$. The $K$-transform of $G$ is the function $KG\colon\Pi(\R^d_0\times \X)\to\R$ defined by:
		\begin{displaymath}
			(K_\Pi G)(\gamma)=(KG)(\gamma):=\sum_{\xi\Subset\gamma}G(\xi),
		\end{displaymath}
		where the inclusion $\xi\Subset\gamma$ means that the sum is taken over all finite subsets of $\gamma$. 
	\end{Definition}
 To avoid any confusion, the dependence on $\Pi$ is omitted. We can see  that by the definition of $B_\mathrm{ls}(\Pi_0(\R^d_0\times \X))$, the $K$-transform is well-defined on such functions.\\
 We recall some results similar to the results which can be found in the theory of homogeneous configuration spaces \cite{MR1914839}.
\begin{Proposition}[\cite{MR1914839}]
\begin{enumerate}
\item The $K$-transform maps functions from $B_\mathrm{ls}(\Pi_0(\R^d_0\times \X))$ into cylinder functions $\mathcal{F}L^0(\Pi(\R^d_0\times \X))$, i.e. for $G\in B_\mathrm{ls}(\Pi_0(\R^d_0\times \X))$, there exists $\Lambda\in\B_c(\R^d_0\times \X)$ such that
			\begin{displaymath}
				(KG)(\gamma)=(KG)(\gamma\cap\Lambda),
			\end{displaymath}
			for all $\gamma\in\Pi(\R^d_0\times \X)$.
			\item The $K$-transform maps $B_\mathrm{bs}(\Pi_0(\R^d_0\times \X))$ to polynomially bounded functions, i.e. for $G\in B_\mathrm{bs}(\Pi_0(\R^d_0\times \X))$, there exist $\Lambda\in\B_c(\R^d_0\times \X), N\in\N$ and $C>0$ such that
			\begin{displaymath}
				|KG|(\gamma)\leq C(1+|\gamma\cap\Lambda|)^N,\ \gamma\in\Pi(\R^d_0\times \X).
			\end{displaymath}
			\item The mapping $K\colon B_\mathrm{ls}(\Pi_0(\R^d_0\times \X))\to\mathcal{F}B^0(\Pi(\R^d_0\times \X))$ is invertible with
			\begin{displaymath}
				K^{-1}F(\gamma)=\sum_{\xi\subset\gamma}(-1)^{|\gamma\setminus\xi|}F(\xi),\ \gamma\in\Pi_0(\R^d_0\times \X).
			\end{displaymath}
			\item $K$ is linear and positivity preserving.
			\item If $G\in B_\mathrm{ls}(\Pi_0(\R^d_0\times \X))$ and continuous, then $KG$ is also continuous.
		\end{enumerate}
	\end{Proposition}
	Let us consider the following example of the $K$-transform of (Lebesgue-Poisson) coherent state corresponding to the function $f$.
 
	\begin{Example}[\cite{MR1914839}]\label{coherent_pi}
		For a function $f\in C_0(\R^d_0\times \X)$, define the coherent state or Lebesgue-Poisson exponent as:
		\begin{displaymath}
			e_\CL(f)\colon\Pi_0(\R^d_0\times \X)\to\R,\ \gamma\mapsto e_\lambda(f,\gamma):=\prod_{(v,x)\in\gamma}f(v,x).
		\end{displaymath}
		Then, $e_\CL(f)\in B_\mathrm{ls}(\Pi_0(\R^d_0\times \X))$. Its $K$-transform is given by:
		\begin{displaymath}
			(Ke_\CL(f))(\gamma)=\prod_{(v,x)\in\gamma}(1+f(v,x)),\ \gamma\in\Pi(\R^d_0\times \X).
		\end{displaymath}
	\end{Example}
	We introduce the $\star$-convolution related to the $K$-transform as the standard convolution on $\R^d$ to the Fourier transform.
	\begin{Definition}\label{def_convolution}
		Let $G_1,G_2\in B_\mathrm{ls}(\Pi_0(\R^d_0\times \X))$. Define the $\star$-convolution as:
		\begin{displaymath}
			(G_1\star G_2)(\gamma):=\sum_{(\xi_1,\xi_2,\xi_3)\in\mathcal{P}^3_\emptyset(\gamma)}G_1(\xi_1\cup\xi_2)G_2(\xi_2\cup\xi_3),\ \gamma\in\Pi_0(\R^d_0\times \X),
		\end{displaymath}
where $\mathcal{P}^3_\emptyset(\gamma)$ denotes all partitions of $\gamma$ into three parts, where the parts may be empty.
	\end{Definition}
	The following relation holds:
	\begin{Proposition}[\cite{MR1914839}]\label{convolution_pi}
		Let $G_1,G_2\in B_\mathrm{ls}(\Pi_0(\R^d_0\times \X))$ be given. Then
		\begin{displaymath}
			K(G_1\star G_2)=KG_1\cdot KG_2.
		\end{displaymath}
	\end{Proposition}
	\subsection{Correlation Measures on $\Pi_0(\R^d_0\times \X)$}
	Our objective is to establish categories of measures on $\Pi_0(\R^d_0\times \X)$ that correspond to probability measures on $\Pi(\R^d_0\times \X)$, utilized to model the state of our system. This approach is based on \cite{MR1914839, MR0323270}. Additionally, we will demonstrate that the group of measures on $\Pi(\R^d_0\times \X)$ with limited local moments allows us to extend the $K$-transform to $L^1$-spaces. To define a measure on $\Pi_0(\R^d_0\times \X)$, we first introduce the pre-kernel $\mathcal{K}$.
	\begin{Definition}\label{prekernel_pi}
		Define the following pre-kernel based on the $K$-transform by:
		\begin{gather}\label{Ktrf_relation}
			\begin{split}
				\mathcal{K}&\colon\B_b(\Pi_0(\R^d_0\times \X))\times\Pi(\R^d_0\times \X)\to[0,\infty)
				\\
				&(A,\gamma)\mapsto\mathcal{K}(A,\gamma):=(K\mathbbm{1}_A)(\gamma).
			\end{split}
		\end{gather}
	\end{Definition}
	We show that $\mathcal{K}$ is in fact a pre-kernel. The property $\mathcal{K}(\emptyset,\gamma)=0$ for any $\gamma\in\Pi_0(\R^d_0\times \X)$ is  obvious. For $\sigma$-additivity, let $A_i\in\B_b(\Pi_0(\R^d_0\times \X)),\ i\in\N$ be disjoint sets such that their countable union also belongs to $\B_b(\Pi_0(\R^d_0\times \X))$. Then there exist $N\in\N$ and $\Lambda\in\B_c(\R^d_0\times \X)$ such that
	\begin{displaymath}
		\bigcup_{i=1}^\infty A_i\subset\bigcup_{k=0}^N\Pi^{(k)}_0(\Lambda).
	\end{displaymath}
	This implies that for $\gamma\in\Pi_0(\R^d_0\times \X)$,
	\begin{displaymath}
		\mathcal{K}\left(\bigcup_{i=1}^\infty A_i,\gamma\right)=\sum_{\xi\Subset\gamma}\sum_{i=1}^\infty\mathbbm{1}_{A_i}(\xi)=\sum_{\substack{\xi\Subset\gamma\\ |\gamma|\leq N}}\sum_{i=1}^\infty\mathbbm{1}_{A_i}(\xi)\sum_{i=1}^\infty\sum_{\substack{\xi\Subset\gamma\\|\xi|\leq N}}\mathbbm{1}_{A_i}(\xi)=\sum_{i=1}^\infty\mathcal{K}(A,\gamma),
	\end{displaymath}
	which completes the proof of the claim. Moreover, $\CK$ can indeed be extended.
	\begin{Lemma}\label{prekernel_ext_pi}
		The pre-kernel $\mathcal{K}$ has a unique extension to a kernel on $\B(\Pi_0(\R^d_0\times \X))\times\Pi(\R^d_0\times \X)$.
	\end{Lemma}
	\begin{proof}
		Since $\B_b(\Pi_0(\R^d_0\times \X))$ is a ring, it suffices to demonstrate the $\sigma$-finiteness of $\mathcal{K}(\cdot,\gamma)$ n order to obtain a unique extension to  $\B(\Pi_0(\R^d_0\times \X))$. For $A\in\B_b(\Pi_0(\R^d_0\times \X))$, the sum
		\begin{displaymath}
			\mathcal{K}(A,\gamma)=\sum_{\xi\Subset\gamma}\mathbbm{1}_A(\xi)
		\end{displaymath}
		is finite. Therefore, by Carath\'eodory's theorem, $\CK$ can be uniquely extended to a kernel on $\B(\Pi_0(\R^d_0\times \X))\times\Pi(\R^d_0\times \X)$.
	\end{proof}
   
	We can further extend Relation \eqref{Ktrf_relation} to more general functions.
	\begin{Proposition}\label{prop_2.52}
		Let $G\colon\Pi_0(\R^d_0\times \X)\to\R$ be a measurable function with $G\geq 0$ or $G\in B_\mathrm{ls}(\Pi_0(\R^d_0\times \X))$. Then,
		\begin{displaymath}
			\int_{\Pi_0(\R^d_0\times \X)}G(\xi)\CK(d\xi,\gamma)=\sum_{\xi\Subset\gamma}G(\xi)=(KG)(\gamma).
		\end{displaymath}
	\end{Proposition}
	\begin{proof}
		The function $G$ can be approximated by a sequence of simple functions, namely:
		\begin{displaymath}
			G(\gamma)=\sum_{k=1}^\infty a_k\mathbbm{1}_{A_k}(\gamma),
		\end{displaymath}
		where $a_k\in\R, A\in\B_b(\Pi_0(\R^d_0\times \X)),\gamma\in\Pi_0(\R^d_0\times \X)$. The identity can then be derived by taking monotone limits. For further details, see \cite{MR1914839, MR0323270}.
	\end{proof}
	We can now construct measures on $\Pi_0(\R^d_0\times \X)$ that correspond to probability measures on $\Pi(\R^d_0\times \X)$ using the kernel $\CK$.
	\begin{Definition}
		Let $\mu$ be a probability measure on $(\Pi(\R^d_0\times \X),\B(\Pi(\R^d_0\times \X)))$. The corresponding correlation measure is defined on $(\Pi_0(\R^d_0\times \X),\B(\Pi_0(\R^d_0\times \X)))$ by the relation:
		\begin{displaymath}
			\rho_\mu(A):=\int_{\Pi(\R^d_0\times \X)}\CK(A,\gamma)\mu(d\gamma).
		\end{displaymath}
	\end{Definition}
    Locally finite correlation measures can be defined through their corresponding probability measures on $\Pi(\R^d_0\times \X)$. However, to guarantee that the correlation measure is indeed locally finite, additional assumptions must be imposed on the measure $\mu$.

\begin{Proposition}\label{finite_moments_pi}
		Let $\mu$ be a probability measure on $(\Pi(\R^d_0\times \X),\B(\Pi(\R^d_0\times \X)))$. Then the corresponding correlation measure $\rho_\mu$ is locally finite if and only if the following holds: for any $\Lambda\in\B_c(\R^d_0\times \X)$ and $N\in\N$,
		\begin{equation}\label{loc_finite}
			\int_{\Pi(\R^d_0\times \X)}|\gamma\cap\Lambda|^N\mu(d\gamma)<\infty.
		\end{equation}
	\end{Proposition}
	\begin{Definition}
		A measure $\mu$ that satisfies property \eqref{loc_finite} is said to have finite local moments of all order. The space of all such measures is denoted by $\CM_\mathrm{fm}^1(\Pi(\R^d_0\times \X))$.
	\end{Definition}
	\begin{proof}[Proof of Proposition \ref{finite_moments_pi}] The proof works analogously to the case of classical configuration spaces,see \cite{MR1914839}.
	\end{proof}
	For the class of measures $\CM^1_\mathrm{fm}(\Pi(\R^d_0\times \X))$, we can extend the $K$-transform to $L^1$-spaces related to these measures.
	\begin{Proposition}[\cite{MR1914839}]
		Let $\mu\in\CM^1_\mathrm{fm}(\Pi(\R^d_0\times \X))$ be given. For all functions $G\in B_\mathrm{bs}(\Pi_0(\R^d_0\times \X))$, we have $G\in L^1(\Pi_0(\R^d_0\times \X),\rho_\mu)$. Furthermore, if $G\geq 0$ or $G\in B_\mathrm{bs}(\Pi_0(\R^d_0\times \X))$, then,
		\begin{equation}\label{eq:K_identity}
			\int_{\Pi_0(\R^d_0\times \X)}G(\ga)\rho_\mu(d\ga)=\int_{\Pi(\R^d_0\times \X)}(KG)(\gamma)\mu(d\gamma).
		\end{equation}
	\end{Proposition}
	\begin{proof}
		The proof follows directly as in \cite{MR1914839}. Since $\mu(\Pi(\R^d_0\times \X))=1$, the restriction from $\Gamma(\R^d_0\times \X)$ to $\Pi(\R^d_0\times \X)$ does not affect the identity.
	\end{proof}
	\begin{Remark}
		For a measure $\mu\in\CM^1_\mathrm{fm}(\Pi(\R^d_0\times \X))$, we may define the correlation measure without using the kernel $\CK$ directly via:
		\begin{displaymath}
			\rho_\mu(A):=\int_{\Pi(\R^d_0\times \X)} K\mathbbm{1}_A(\gamma)\mu(d\gamma),\ A\in\B_b(\Pi_0(\R^d_0\times \X)).
		\end{displaymath}
		This follows from Proposition \ref{finite_moments_pi}, since $K\mathbbm{1}_A\in L^1(\mu)$ for $A\in\B_b(\Pi_0(\R^d_0\times \X))$.
	\end{Remark}
	\begin{Definition}
		The remark above enables us to  define the dual operator of $K$, i.e.
		\begin{align*}
			K^*\colon&\CM^1_\mathrm{fm}(\Pi(\R^d_0\times \X))\to\CM_\mathrm{lf}(\Pi_0(\R^d_0\times \X))
			\\
			&\mu\mapsto K^*\mu:=\rho_\mu.
		\end{align*}
	\end{Definition}
	
To complete the extension of the 
$K$-transform, we require one final continuity result for this mapping.
	\begin{Lemma}[\cite{MR1914839}]
		Let $\{G_n\}_{n\in\N}\subset B_\mathrm{bs}(\Pi_0(\R^d_0\times \X))$ be a sequence which converges in \, $L^1(\Pi_0(\R^d_0\times \X),\rho_\mu)$ for some measure $\mu\in\CM^1_\mathrm{fm}(\Pi(\R^d_0\times \X))$. Then $\{KG_n\}_{n\in\N}$ converges in $L^1(\Pi(\R^d_0\times \X),\mu)$.
	\end{Lemma}
	\begin{proof}
		The proof follows by applying the triangle inequality: 
 $|KG|\leq K|G|$.
	\end{proof}
	We can now prove the extension result for the  $K$-transform on $L^1$-spaces.
	\begin{Theorem}[\cite{MR1914839}]\label{Ktrf_extension}
		Let $\mu\in\CM_\mathrm{fm}^1(\Pi(\R^d_0\times \X))$ be given. For any $G\in L^1(\Pi_0(\R^d_0\times \X),\rho_\mu)$, we define
		\begin{displaymath}
			KG(\gamma):=\sum_{\xi\Subset\gamma}G(\xi),
		\end{displaymath}
where the series converges absolutely $\mu$-almost surely. Furthermore, we have the following estimate:
		\begin{displaymath}
			\|KG\|_{L^1(\mu)}\leq\|K|G|\|_{L^1(\mu)}=\|G\|_{L^1(\rho_\mu)},
		\end{displaymath}
which implies that $KG\in L^1(\mu)$ and for all $G\in L^1(\rho_\mu)$,
\begin{equation}
			\int_{\Pi_0(\R^d_0\times \X)}G(\ga)\rho_\mu(d\ga)=\int_{\Pi(\R^d_0\times \X)}(KG)(\gamma)\mu(d\gamma).
\end{equation}
  
	\end{Theorem}
	\begin{proof}
    For non-negative functions, the result is derived through computations involving the preceding lemma and Fatou's lemma. The general case is addressed by decomposing the function into its positive and negative components.
		%The result for non-negative functions follows from calculations using the previous lemma and Fatou's lemma. The extension to general functions is achieved by decomposing the function into its positive and negative parts.
	\end{proof}
The identity given by Proposition \ref{convolution_pi} can be extended, considering some conditions.
	\begin{Proposition}[\cite{MR1914839}]
		Let $G_1,G_2\in B_\mathrm{ls}(\Pi_0(\R^d_0\times \X))$ and let $\mu\in\CM^1_\mathrm{fm}(\Pi(\R^d_0\times \X))$.Then the following identity holds $\mu$-almost surely:
		\begin{displaymath}
			K(G_1\star G_2)=KG_1\cdot KG_2,
		\end{displaymath} if one of the following conditions is satisfied:
		\begin{enumerate}
			\item $G_1,G_2\geq 0$;
			\item $|G_1|\star|G_2|\in L^1(\rho_\mu)$ (and consequentially $K(G_1\star G_2)\in L^1(\mu)$);
			\item $G_1,G_2\in L^1(\rho_\mu)$.
		\end{enumerate}
		
	\end{Proposition}
	\begin{proof}
		A direct consequence of the Theorem \ref{Ktrf_extension}.
	\end{proof}

\subsection{Correlation Functions on $\Pi_0(\R^d_0\times \X)$}\label{ch_corrfn_pi}
 Our goal is to establish the existence of a correlation function associated with a specific subset of measures, denoted by \(\CM_\mathrm{fm}^1(\Pi(\R^d_0 \times \X))\). This correlation function acts as the density function for the corresponding correlation measure \(\rho_\mu\). Such functions are of particular mathematical interest due to their widespread use in applications where they characterize the behavior of the system under consideration.
%	We aim at proving the presence of a correlation function that corresponds to a specific subset of measures denoted by $\CM_\mathrm{fm}^1(\Pi(\R^d_0\times \X))$. This correlation function serves as the density function for the correlation measure $\rho_\mu$. These functions hold significant interest as they are extensively employed in applications to describe the behaviour of the system they pertain to.
	\begin{Definition}
		A measure $\mu\in\CM^1_\mathrm{fm}(\Pi(\R^d_0\times \X))$ is locally absolutely continuous with respect to the Poisson measure $\pi_\sigma$ iff the measure $\mu^\Lambda$ is absolutely continuous with respect to $\pi_\sigma^\Lambda$ for all $\Lambda\in\B_c(\R^d_0\times \X)$, where $\mu^\Lambda:=\mu\circ p^{-1}_\Lambda$.
	\end{Definition}
	
	\begin{Proposition}[\cite{MR1914839}]\label{density_on_pi}
		For a measure $\mu\in\CM^1_\mathrm{fm}(\Pi(\R^d_0\times \X))$ which is locally absolutely continuous  with respect to $\pi_\sigma$, the correlation measure $\rho_\mu$ is absolutely continuous  with respect to the Lebesgue-Poisson measure $\CL_\sigma$ introduced in Section 2.1. The density function has the following representation for any $\gamma\in\Pi_0(\Lambda)$:
		\begin{displaymath}
			k_\mu(\gamma)=\frac{d\rho_\mu}{d\CL_\sigma}(\gamma)=\int_{\Pi(\Lambda)}\frac{d\mu^\Lambda}{d\pi^\Lambda_\sigma}(\gamma\cup\xi)\pi_\sigma^\Lambda(d\xi).
		\end{displaymath}
	\end{Proposition}
	\begin{Definition}\label{def_corrfn_pi}
		The function $k_\mu\colon\Pi_0(\R^d_0\times \X)\to\R$ defined by the previous proposition is called the correlation function corresponding to $\mu$. Moreover, we have the decomposition $k_\mu\simeq\{k_\mu^{(n)}\}_{n=0}^\infty$, where for any $n\in\N, k_\mu^{(n)}\colon(\R^d_0\times \X)^n\to\R$ is a symmetric function with
		\begin{equation*}
			k_\mu^{(n)}(v_1,x_1,\dotsc,v_n,x_n):=
			\begin{cases}
				k_\mu(\{(v_1,x_1),\dotsc,(v_n,x_n)\}),&\text{ if }|\{(v_1,x_1),\dotsc,(v_n,x_n)\}|=n,
				\\
				0,&\text{ otherwise}.
			\end{cases}
		\end{equation*}
		The functions $k_\mu^{(n)}$ are called $n$-point correlation functions.
	\end{Definition}
	In what follows, we present the Bogoliubov functional, which we will employ to define correlation functions on the cone $\K(\X)$. Bogoliubov originally developed this category of functionals \cite{B} to establish correlation functions for systems in statistical mechanics. The Bogoliubov functional approach also investigates continuum interacting particle systems in \cite{MR2253724}.

	\begin{Definition}
		Let $\mu\in\CM^1_\mathrm{fm}(\Pi(\R^d_0\times \X))$. The Bogoliubov functional $L^\Pi_\mu$ corresponding to $\mu$ is a functional defined at each measurable function $\varphi\colon\R^d_0\times \X\to\R$ by
		\begin{displaymath}
			L^\Pi_\mu(\varphi):=\int_{\Pi(\R^d_0\times \X)}\prod_{(v,x)\in\gamma}(1+\varphi(v,x))\mu(d\gamma),
		\end{displaymath}
		provided, the right-hand side exists for $|\varphi|$.
	\end{Definition}
	The following proposition will be useful to define the correlation functions on $\K(\X)$ later.

	\begin{Proposition}[\cite{MR2253724}]\label{prop_bogo_pi}
		Under some assumptions, the Bogoliubov functional is the generating functional of the correlation function. In other words, for any $\varphi\colon\R^d_0\times \X\to\R$ such that $L_\mu^\Pi(\varphi)$ is well-defined, we have
		\begin{align*}
			L_\mu^\Pi(\varphi)&=\sum_{n=0}^\infty\frac{1}{n!}\int_{(\R^d_0\times \X)^n}\varphi(v_1,x_1)\dotsb\varphi(v_n,x_n)\times
			\\
			&\hspace{100pt}\vphantom{\int}\times k_\mu^{(n)}(v_1,\dotsc,x_n)\lambda(dv_1)m(dx_1)\dotso\lambda(dv_n)m(dx_n).
		\end{align*}
	\end{Proposition}

	\section{Harmonic Analysis on $K(\X)$}\label{harmana_k}
	We aim to present harmonic analysis on the cone of vector-valued discrete Radon measures on $\X$, denoted by $\K(\X)$, after introducing harmonic analysis on the Plato space $\Pi(\R_0^d \times \X)$. To establish a relationship between $\Pi(\R_0^d \times \X)$ and $\K(\X)$, we employ the reflection mapping $\CR$.
    
\subsection{The $K$-Transform}

Initially, we focus on the space $\K_0(\X)$ and, similar to the previous section, introduce subspaces of $\K_0(\X)$ to decompose the space. Subsequently, we proceed to introduce classes of functions in these spaces.

\begin{Definition}
\begin{enumerate}
\item The set of discrete Radon measures with finite support is defined as:
			\begin{displaymath}
				\K_0(\X):=\left\{\eta\in\K(\X)\colon|\tau(\eta)|<\infty\right\}.
			\end{displaymath}
\item For $n\in \N_0$, the set of $n$-points measures is defined as:
			\begin{displaymath}
				\K_0^{(n)}(\X):=\left\{\eta\in\K_0(\X)\colon |\tau(\eta)|=n\right\},\ n\in\N
			\end{displaymath}
			and $\K_0^{(0)}(\X)=\{0\}$ the set consisting of the zero measure.
\item For a compact set $\Lambda\subset\X$, the set of all measures supported in $\Lambda$ is defined as:
			\begin{displaymath}
				\K_0(\Lambda):=\left\{\eta\in\K_0(\X)\colon\tau(\eta)\subset\Lambda\right\}.
			\end{displaymath}
\item A set $A\subset\K_0(\X)$ is called bounded if there exists a compact set $\Lambda\subset\X$ and $N\in\N$ such that
			\begin{displaymath}
				A\subset\bigcup_{n=0}^N\K^{(n)}_0(\Lambda).
			\end{displaymath}
			Denote the collection of all bounded  Borel subsets of $\K_0(\X)$ by $\B_b(\K_0(\X))$.
			\item A bounded set $A\subset\K_0(\X)$ is said to have compact velocities if, additionally, there exists a compact set $I\subset
			\R_0^d $ such that
			\begin{displaymath}
				A\cap\{\eta\in\K_0(\X)\mid\exists x\in\tau(\eta)\colon v_x\notin I\}=\emptyset.
			\end{displaymath}
			Denote the collection of all such sets by $\B_\mathrm{cm}(\K_0(\X))$.
		\end{enumerate}
		Note that we have
		$$
		\K_0(\X)=\bigsqcup_{n=0}^\infty   \K_0^{(n)}(\X)
		$$
        and 
		$$
		\K_0(\X)=\bigcup_{\Lambda\in\B_c(\X)}\K_0(\Lambda),
		$$
		where the first union is disjoint.
	\end{Definition}
	
By reflection mapping $\CR$ we relate the subspaces of $\Pi(\R^d_0\times\X)$ and $\K(\X)$.

	\begin{Proposition}
		The following relations hold:
		\begin{enumerate}
			\item $\CR\Pi_0(\R^d_0\times \X)=\K_0(\X)$.
			\item $\CR\Pi_0^{(n)}(\R^d_0\times \X)=\K_0^{(n)}(\X)$ for any $n\in\N_0$.
			\item $\CR\Pi_0(\R^d_0\times\Lambda)=\K_0(\Lambda)$ for any set $\Lambda\subset\X$.
			\item For any $A\in\B_b(\Pi_0(\R^d_0\times \X))$, we have $\CR A\in\B_\mathrm{cm}(\K_0(\X))$ and vice versa.
		\end{enumerate}
	\end{Proposition}
	\begin{proof}
		We will now prove the first statement. The other statements follow similarly.
		
		For $\gamma\in\Pi_0(\R^d_0\times \X)$, there is a representation $\gamma=\sum_{i=1}^n\delta_{(v_i,x_i)}$. This implies that  $\CR\gamma=\sum_{i=1}^nv_i\delta_{x_i}\in\K_0(\X)$.
		
		Let $\eta\in\K_0(\X)$. Again, we can represent it as $\eta=\sum_{i=1}^nv_i\delta_{x_i}$. By defining $\gamma=\sum_{i=1}^n\delta_{(v_i,x_i)}$, we obtain $\gamma\in\Pi_0(\R^d_0\times \X)$ and $\CR\gamma=\eta$.
	\end{proof}

	We now introduce the corresponding function spaces on  $\K_0(\X)$. 
	
	\begin{Definition}\label{locsupp}
		\begin{enumerate}
			\item  A function $G\colon\K_0(\X)\to\R$ is said to be bounded with local support if there exist $C>0$ and $\Lambda\in\B_c(\R^d)$ such that the following estimate holds for all $\eta\in\K_0(\X)$:
			\begin{equation}\label{locbd}
				|G(\eta)|\leq C\mathbbm{1}_{\K_0(\Lambda)}(\eta)\prod_{x\in\tau(\eta)}|v_x|.
			\end{equation}
Note that, this implies that $G(\eta)=0$ if $\tau(\eta)\cap\Lambda^c\neq\emptyset$. We denote by $B_\mathrm{ls}(\K_0(\X))$ all measurable functions $G\colon\K_0(\X)\to\R$ which are bounded with local support.
			\item A function $G\colon\K_0(\X)\to\R$ is called bounded with bounded support if there exist $\Lambda\in\mathcal{B}_c(\R^d), N\in\N$ and $C>0$ such that
			\begin{displaymath}
				|G(\eta)|\leq C\mathbbm{1}_{\K_0(\Lambda)}(\eta)\mathbbm{1}_{\{|\tau(\eta)|\leq N\}}(\eta), \end{displaymath}
			i.e. $G(\eta)=0$ whenever $|\tau(\eta)|>N$ or $\tau(\eta)\cap\Lambda^c\neq\emptyset$. Denote the space of all such functions by $B_\mathrm{bs}(\K_0(\X))$.
			\item Taking into account the effect of the marks as above, we define a modified version of $B_\mathrm{bs}(\K_0(\X))$. For some $\Lambda\in\mathcal{B}_c(\R^d), N\in\N$ and $C>0$, we define the space $\widetilde{B}_\mathrm{bs}(\K_0(\X))$ as a space of all functions $G\colon\K_0(\X)\to\R$ which satisfy the bound
			\begin{displaymath}
				|G(\eta)|\leq C\mathbbm{1}_{\K_0(\Lambda)}(\eta)\mathbbm{1}_{\{|\tau(\eta)|\leq N\}}(\eta)\prod_{x\in\tau(\eta)}|v_x|.
			\end{displaymath}
			Obviously, we have $\widetilde{B}_\mathrm{bs}(\K_0(\X))\subset B_\mathrm{ls}(\K_0(\X))$.

\item Define the space of bounded measurable functions with compact mark support as the set of all functions $G\in B_\mathrm{bs}(\K_0(\X))$ such that there exists a compact set $I\subset \R^d_0 $ for which
			\begin{equation}\label{comp_marks}
				|G(\eta)|\leq C\mathbbm{1}_{\K_0(\Lambda)}\mathbbm{1}_{\{|\tau(\eta)|\leq N\}}\prod_{x\in\tau(\eta)}\mathbbm{1}_I(v_x),
			\end{equation}
			where $\Lambda, C$ and $N$ are as above.  Denote the space of bounded functions with compact marks by $B_\mathrm{cm}(\K_0(\X))$.
			\item A measure $\rho$ on $\K_0(\X)$ is called locally finite if for any $\Lambda\in\B_c(\R^d)$ and for any $m\in\N_0$, the value of $\rho(\K_0^{(m)}(\Lambda))$ is finite. Equivalently, $\rho(A)$ is finite for all bounded measurable sets $A\subset\K_0(\X)$. The space of all locally finite measures on $\K_0(\X)$ is denoted by $\CM_\mathrm{lf}(\K_0(\X))$.
			\item A measure $\rho$ on $\K_0(\X)$ is called mark-locally finite if $\rho(A)<\infty$ for all $A\in\B_\mathrm{cm}(\K_0(\X))$. A locally finite measure $\rho$ is also mark-locally finite.
		\end{enumerate}
	\end{Definition}
	
We now establish the relationship between the function spaces on $\Pi_0(\R^d_0 \times \X)$ and  $\K_0(\X)$. Using the reflection mapping $\CR$, we map functions in 
 $\CF(\Pi_0(\R^d_0\times \X))$ , defined on  $\Pi_0(\R^d_0\times \X)$ to functions in $\CF(\K_0(\X))$, defined on $\K_0(\X)$ , as follows
	\begin{Definition}
		Define the pushforward of functions on $\Pi_0(\R^d_0\times \X)$ to $\K_0(\X)$ as follows:
		\begin{align*}
			\CR\colon&\CF(\Pi_0(\R^d_0\times \X))\to\CF(\K_0(\X))
			\\
			&F\mapsto\CR F:=F\circ\CR^{-1},
		\end{align*}
		analogously, we may define the inverse mapping $\CR^{-1}\colon\CF(\K_0(\X))\to\CF(\Pi_0(\R^d_0\times \X))$.
	\end{Definition}

	While for the function spaces on $\Pi_0(\R^d_0\times \X)$ we require compactness of mark support, for the function spaces on $\K_0(\X)$  according to the Definition \ref{locsupp} we have boundedness in the mark variables.
	The following proposition shows the relations between locally supported functions on $\Pi_0(\R^d_0\times \X)$ and $\K_0(\X)$.
	
	\begin{Proposition}\label{relation_functions_pi_k}
		For the above spaces, the following relations hold:
		\begin{enumerate}
			\item $\CR B_\mathrm{ls}(\Pi_0(\R^d_0\times \X))\not\subset B_\mathrm{ls}(\K_0(\X))$ and $B_\mathrm{ls}(\Pi_0(\R^d_0\times \X))\not\supset\CR^{-1} B_\mathrm{ls}(\K_0(\X))$;
			\item $\CR B_\mathrm{bs}(\Pi_0(\R^d_0\times \X))=B_\mathrm{cm}(K_0(\X))$.
		\end{enumerate}
	\end{Proposition}
	\begin{proof}
		\begin{enumerate}
			\item Let $G\in B_\mathrm{ls}(\Pi_0(\R^d_0\times \X))$ such that for some compact $A\subset \R^d_0$, we have $A\times\Lambda^\prime\subset\Lambda$, where $\Lambda$ is as in Definition \ref{locsupp_pi} and $\Lambda^\prime\subset\X$ compact. We require the estimate
			\begin{displaymath}
				C\mathbbm{1}_{\Pi_0(\Lambda)}(\CR^{-1}\eta)\leq C_1\mathbbm{1}_{\K_0(\Lambda^\prime)}(\eta)\prod_{x\in\tau(\eta)}v_x
			\end{displaymath}
			for some $C,C_1>0$. But since $A$ is compact is possible, and the number of points in $\eta$ is arbitrary, the right-hand side can be arbitrarily small. On the other hand, let $G\in B_\mathrm{ls}(\K_0(\X))$. To show $\CR^{-1}G\in B_\mathrm{ls}(\Pi_0(\R^d_0\times \X))$, we require
			\begin{displaymath}
				C\mathbbm{1}_{\K_0(\Lambda^\prime)}(\CR\gamma)\prod_{x\in\tau(\CR\gamma)}v_x\leq C_1\mathbbm{1}_{\Pi_0(\Lambda)}(\CR\gamma)
			\end{displaymath}
			for some $C,C_1>0$ and $\Lambda,\Lambda^\prime$ as in the definitions above. Since there is no compactness requirement on the marks in $B_\mathrm{ls}(\K_0(\X))$, the left-hand side can be arbitrarily large.
			
		\item	Let $G\in B_\mathrm{bs}(\Pi_0(\R^d_0\times \X))$. Then there exist
			$\La \in \B_c (\R^d_0\times\X)$ 
			compact, $N\in\N$ and $C>0$ such that  \eqref{bbs}  holds. Then, there exists $A$ such that  $\La\subset A\times\Lambda^\prime$ for some $\Lambda^\prime\in\B_c(\X)$. Then,
			$$
			G(\CR^{-1}\eta)
			%\leq C\mathbbm{1}_{\Pi_0(\Lambda)}(\CR^{-1}\eta)\mathbbm{1}_{\{|\CR^{-1}\eta|\leq N\}}(\R^{-1}\eta)
			\leq C\mathbbm{1}_{\Pi_0([A\times\Lambda^\prime)}(\CR^{-1}\eta)\mathbbm{1}_{\{|\CR^{-1}\eta|\leq N\}}(\CR^{-1}\eta)
			$$
			$$
			=C\mathbbm{1}_{\K_0(\Lambda^\prime)}(\eta)\mathbbm{1}_{\{|\tau(\eta)|\leq N\}}(\eta)\prod_{x\in\tau(\eta)}\mathbbm{1}_A (v_x),
			$$
			
			which shows the first inclusion. On the other hand, let $G\in B_\mathrm{cm}(\K_0(\X))$. Then, there exist $\Lambda^\prime\in\B_c(\X)$, $I\in\B_c(\R^d_0)$, $N\in\N$ and $C>0$ such that \eqref{comp_marks} holds. Also, $I\times\Lambda^\prime\in\B_c(\R^d_0 \times\X)$ and since
			\begin{displaymath}
				\mathbbm{1}_{\K_0(\Lambda^\prime)}(\eta)\prod_{x\in\tau(\eta)}\mathbbm{1}_I(v_x)=\mathbbm{1}_{\Pi_0(I\times\Lambda^\prime)}(\CR^{-1}\eta),
			\end{displaymath}
			the claim follows.
		\end{enumerate}
	\end{proof}
	
	We  define the $K$-transform analogously to the case of space $\Pi_0(\R^d_0\times\X)$. We note that an estimate for the $K$-transform is obtained using the definition of function spaces \eqref{locsupp}.
\begin{Definition}
Let $G\in B_\mathrm{ls}(\K_0(\X))$. The $K$-transform of $G$ is defined as the function $KG\colon\K(\X)\to\R$ of the form:
		\begin{displaymath}
			(K_\K G)(\eta)=(KG)(\eta):=\sum_{\xi\Subset\eta}G(\xi),
		\end{displaymath}
where the inclusion $\xi\Subset\eta$ is meant in the sense of Definition \ref{K transform}. 
	\end{Definition}
	
	\begin{Lemma}
		For any $G\in B_\mathrm{ls}(\K_0(\X))$, the $K$-transform is well-defined and the following estimate holds:
		\begin{displaymath}
			|(KG)(\eta)|\leq C\prod_{x\in\tau(\eta)\cap\Lambda}(1+|v_x|),
		\end{displaymath}
where $C$ and $\Lambda$ are as in Definition \ref{locsupp}.
	\end{Lemma}
	\begin{proof}
		We have
		\begin{displaymath}
			|(KG)(\eta)|\leq\sum_{\xi\Subset\eta}|G(\eta)|\leq C\sum_{\substack{\xi\in \K_0(\Lambda)\\ \tau(\xi)\subset\tau(\eta)}}\prod_{x\in\tau(\xi)}|v_x|=C\prod_{x\in\tau(\eta)\cap\Lambda}(1+|v_x|),
		\end{displaymath}
		where the product in the last expression is finite if and only if the following sum is finite:
		\begin{displaymath}
			\sum_{x\in\tau(\eta)\cap\Lambda}|v_x|.
		\end{displaymath}
		 Since the latter holds by the definition of 
 $\eta\in\K(\X)$, the claim follows.
	\end{proof}
	
	Similar to the example in the previous section, we present $\K_0(\X)$-analogue of coherent states.
	\begin{Example}
		For a function $\varphi\in C_0(\R^d)$ and a vector $h\in \X$ we define the coherent state as the function $e^h_\K(\varphi)\colon\K_0(\X)\to\R$ by:
		\begin{displaymath}
			e_{\K}(\varphi,\eta):=\prod_{x\in\tau(\eta)} <h,v_x>\varphi(x),
		\end{displaymath}
		since $\varphi$ is bounded, $e^h_\K(\varphi)$ fulfills bound \eqref{locbd}. We can calculate its $K$-transform:
		\begin{displaymath}
			(Ke^h_\K(\varphi))(\eta)=\prod_{x\in\tau(\eta)}(1+<h,v_x>\varphi(x)).
		\end{displaymath}
		For the right-hand-side to be well-defined, the series $\sum_{x\in\tau(\eta)}<h,v_x>\varphi(x)$ needs to be convergent. This is given in our case since $\varphi$ is compactly supported, this is given in our case.
		
		For $f^h_\varphi(v,x):= <h,v>\varphi(x)$, consider the Lebesgue-Poisson exponent $e_\CL(f^h_\varphi)$ from Example \ref{coherent_pi}, we see that
		\begin{displaymath}
			e^h_\K(\varphi,\CR\gamma)=e_\CL(f^h_\varphi,\gamma),\ \gamma\in\Pi_0(\R^d_0\times\X).
		\end{displaymath}
	\end{Example}

	We can relate the $K$-transform on $\Pi_0(\R^d_0\times \X)$ and on $\K_0(\X)$ in the following way:
	\begin{Proposition}\label{ktrf_rel_pi_k}
		For $G\in B_\mathrm{ls}(\K_0(\X))\cap\CR B_\mathrm{ls}(\Pi_0(\R^d_0\times \X))$ and $\eta\in\K(\X)$, the following holds:
		\begin{displaymath}
			(K_\K G)(\eta)=(K_\Pi(\CR^{-1}G))(\CR^{-1}\eta).
		\end{displaymath}
	\end{Proposition}
	
	\begin{proof}[Proof of Proposition \ref{ktrf_rel_pi_k}]
		Let $\eta=\sum_{i\in I}v_i\delta_{x_i}$, where $I\subset\N$. Then,
		\begin{align*}
			(K_\K G)(\eta)&=\sum_{n=0}^\infty\sum_{\{i_1,\dotsc,i_n\}\subset I}G\left(\sum_{i=1}^n v_{i_k}\delta_{x_{i_k}}\right)=\sum_{n=0}^\infty\sum_{\{i_1,\dotsc,i_n\}\subset I}G\left(\CR\left[\sum_{i=1}^n\delta_{(x_{i_k},v_{i_k})}\right]\right)
			\\
			&=\sum_{n=0}^\infty\sum_{\{i_1,\dotsc,i_n\}\subset I}(\CR^{-1}G)\left(\sum_{i=1}^n\delta_{(x_{i_k},v_{i_k})}\right)
			\\
			&=\sum_{n=0}^\infty\sum_{\{i_1,\dotsc,i_n\}\subset I}(\CR^{-1}G)\left(\CR^{-1}\left[\sum_{i=1}^n v_{i_k}\delta_{x_{i_k}}\right]\right)
			\\
			&=(K_\Pi (\CR^{-1}G))(\CR^{-1}\eta).
		\end{align*}
	\end{proof}

	We need properties related to the $K$-transform, which will be used in the calculation below.
	
	\begin{Lemma}\label{combi}
		Let $G,G_1,G_2\in B_\mathrm{ls}(\K_0(\X))$.
		\begin{enumerate}
			\item The K-transform has the following properties:
			\begin{align*}
				KG(\eta-v_x\delta_x)-KG(\eta)&=-(KG(\cdot+v_x\delta_x))(\eta-v_x\delta_x),
				\\
				KG(\eta+v_x\delta_x)-KG(\eta)&=(KG(\cdot+v_x\delta_x))(\eta).
			\end{align*}
			\item The K-transform and the $\star$-convolution have the following relation:
			\begin{displaymath}
				K(G_1\star G_2)=KG_1\cdot KG_2.
			\end{displaymath}
		\end{enumerate}
	\end{Lemma}
	\begin{proof}
		Proof can be conducted using combinatorial arguments \cite{MR1914839}.
	\end{proof}
   The following lemma is needed for calculations on the space of finite measures. It is also known as  Minlos Lemma.
	\begin{Lemma}[\cite{MR2426716}]\label{Minlos}
		Let $\CL_\sigma$ be the Lebesgue-Poisson measure on $\K_0(\X)$ associated with some intensity measure 
		$\sigma =\la\otimes m$.
		\begin{enumerate}
			\item Let $G\colon\K_0(\X)\to\R$, $H\colon(\K_0(\X))^2\to\R$. Then,
			\begin{align*}
				\int_{\K_0(\X)}&\int_{\K_0(\X)}G(\xi_1+\xi_2)H(\xi_1,\xi_2) \CL_\sigma(d\xi_1)\CL_\sigma(d\xi_2)
				\\
				&=\int_{\K_0(\X)}G(\eta)\sum_{\xi\subset\eta}H(\xi,\eta-\xi)\CL_\sigma(d\eta).
			\end{align*}
			\item Let $H\colon\K_0(\X)\times\R^d_0\times\R^d\to\R$. Then,
			\begin{align*}
				\int_{\K_0(\X)}&\sum_{x\in\tau(\eta)}H(\eta,v_x,x)\CL_\sigma(d\eta)
				\\
				&=\int_{\K_0(\X)}\int_{\R^d_0\times\R^d} H(\eta+v\delta_x )\sigma(dv,dx)\CL_\sigma(d\eta),
			\end{align*}
		\end{enumerate}
		provided that at least one side of the equation exists.
	\end{Lemma}
	
	We aim to extend the $K$-transform to the whole space $L^1(\rho)$ for (mark-)locally finite measures $\rho$ on $\K_0(\X)$.
    For this purpose, we show that $B_\mathrm{cm}(\K_0(\X))$ is dense in  $L^1$-spaces. Similarly, we prove that $\widetilde{B}_\mathrm{bs} (\K_0(\X))$ is dense in $L^1$ with a modified measure.\par
    By Proposition \ref{relation_functions_pi_k} we have $B_\mathrm{cm}(\K_0(\X))\subset B_\mathrm{ls}(\K_0(\X))\cap\CR B_\mathrm{ls}(\Pi_0(\R^d_0\times \X))$. Hence, the relation in Proposition \ref{ktrf_rel_pi_k} holds on a set of functions in $L^1(\K_0,\rho)$ for a class of measures $\rho$.
	\begin{Lemma}\label{bsdense}
		For any locally finite measure $\rho$, the space $B_\mathrm{bs}(\K_0(\X))$ is dense in $L^1(\rho)$.
	\end{Lemma}
	\begin{proof}
		Let $G\in L^1(\K_0(\X),\rho)$ for some measure $\rho$ on $\K_0(\X)$. We begin by approximating unbounded functions with bounded support. Define
		\begin{align*}
			G_n(\eta)&:=\left[G(\eta)\mathbbm{1}_{\K_0(B_n)}(\eta)\mathbbm{1}_{\{|\tau(\eta)|\leq n\}}(\eta)\right]\land n
			\\
			G^\prime_n(\eta)&:=\left[G(\eta)\mathbbm{1}_{\K_0(B_n)}(\eta)\mathbbm{1}_{\{|\tau(\eta)|\leq n\}}(\eta)\right],
		\end{align*}
		where $B_n\subset\R^d$ is the ball with radius $n$ centered at $0$. Then $G_n\in B_\mathrm{bs}(\K_0(\X))$ and
		\begin{align*}
			\|G_n(\eta)-G^\prime_n(\eta)\|_{L^1(\rho)}&=\int_{\K_0(\X)}|G_n(\eta)-G^\prime_n(\eta)|\rho(d\eta)
			\\
			&=\int_{\K_0(\X)}|G_n^\prime(\eta)|\mathbbm{1}_{\{|G_n(\eta)|\geq n\}}(\eta)\rho(d\eta)
			\\
			&\leq\int_{\K_0(\X)}|G(\eta)|\mathbbm{1}_{\{|G(\eta)|\geq n\}}(\eta)\rho(d\eta).
		\end{align*}
		Since $G\in L^1(\K_0(\X),\rho)$, the last term converges to $0$ for $n\to\infty$. Next, define
		\begin{displaymath}
			G^{\prime\prime}_n(\eta)=G(\eta)\mathbbm{1}_{\K_0(B_n)}(\eta).
		\end{displaymath}
Recall that $\K_0(\X)$ can be decomposed into disjoint n-point configurations, i.e.,
		\begin{displaymath}
			\K_0(\X)=\bigcup_{m=0}^\infty\K_0^{(m)}(\X),\ \K_0^{(m)}(\X)=\left\{\eta\in\K_0(\X)\colon|\tau(\eta)|=m\right\}.
		\end{displaymath}
Using this decomposition, we get
		\begin{align*}
			\|G^\prime_n(\eta)-G^{\prime\prime}_n(\eta)\|&=\int_{\K_0(\X)}|G^\prime_n(\eta)-G^{\prime\prime}_n(\eta)|\rho(d\eta)
			\\
			&=\sum_{m=n+1}^\infty\int_{\K_0^{(n)}(\X)}|G^{\prime\prime}_n(\eta)|\rho(d\eta)\leq\sum_{m=n+1}^\infty\int_{\K_0^{(n)}(\X)}|G(\eta)|\rho(d\eta)
		\end{align*}
		and since $G\in L^1(\K_0(\X))$, the series is absolutely convergent. Therefore,  the last expression tends to $0$ as $n\to\infty$. For the final step, observe that the increasing sequence $\{\K_0(B_n)\}_{n=1}^\infty$ approximates $\K_0(\X)$, and thus
		\begin{displaymath}
			\|G^{\prime\prime}_n-G\|_{L^1}=\int_{\K_0(\X)}|G(\eta)|\rho(d\eta)-\int_{\K_0(B_n)}|G(\eta)|\rho(d\eta),
		\end{displaymath}
		which converges to $0$ as $n\to\infty$ by the argument provided above.
	\end{proof}
	
	\begin{Corollary}\label{cordense}
		Define the density function:
		\begin{displaymath}
			f(\eta)=\prod_{x\in\tau(\eta)}\frac{1}{|v_x|}.
		\end{displaymath}
		Then the space $\widetilde{B}_\mathrm{bs}(\K_0(\X))$ is dense in $L^1(\K_0(\X),f\rho)$.
	\end{Corollary}
	\begin{proof}
		Let $G\in L^1(f\rho)$. By definition, we have
		\begin{displaymath}
			\|G\|_{L^1(f\rho)}=\int_{\K_0(\X)}|G(\eta)|\prod_{x\in\tau(\eta)}\frac{1}{|v_x|}\rho(d\eta)<\infty,
		\end{displaymath}
		which implies that $G\cdot f\in L^1(\rho)$. Since $B_\mathrm{bs}(\K_0(\X))$ is dense in $L^1(\rho)$, there exists a sequence $\{G_n\}_{n=1}^\infty\subset B_\mathrm{bs}(\K_0(\X))$ such that
		\begin{displaymath}
			\|G_n-G\cdot f\|_{L^1(\rho)}\to 0,\ n\to\infty.
		\end{displaymath}
		On the other hand, the sequence $\{\widetilde{G}_n\}_{n=1}^\infty$ is in $\widetilde{B}_\mathrm{bs}(\K_0(\X))$, where $\widetilde{G}_n:=\frac{G_n}{f}$ and satisfies
		\begin{displaymath}
			\widetilde{G}_n(\eta)=\frac{G_n(\eta)}{f(\eta)}\leq C\mathbbm{1}_{\K_0(\Lambda)}(\eta)\mathbbm{1}_{\{|\tau(\eta)|\leq N\}}(\eta)\prod_{x\in\tau(\eta)}v_x.
		\end{displaymath}
Moreover, $\widetilde{G}_n$ converges to $G$ in $L^1(f\rho)$:
		\begin{align*}
			\|\widetilde{G}_n-G\|_{L^1(f\rho)}&=\int_{\K_0(\X)}\left|\frac{G_n}{f}-G\right|fd\rho
			\\
			&=\int_{\K_0(\X)}|G_n-Gf|d\rho=\|G_n-Gf\|_{L^1(\rho)}\to0\ n\to\infty.
		\end{align*}
		This completes the proof.
	\end{proof}
	
	The following is an example of a locally finite measure on $\K_0(\X)$, derived from the Lebesgue-Poisson measure on $\Pi_0(\R^d_0\times \X)$. 
	\begin{Example}\label{exlf}
		Let $\la$ be a locally finite measure on $\R^d_0$ and $m$ a non-atomic measure on $\X$ (e.g. the Lebesgue measure). Define the measure $\CL=\CL_{\la \otimes m}$ as:
		\begin{align*}
			\int_{\K_0(\X)}&F(\eta)\CL (d\eta)=
			\\
			&=F(0)+\sum_{n=1}^\infty\frac{1}{n!}\int_{(\R^d_0\times\R^d)^n}F\left(\sum_{i=1}^n v_i\delta_{x_i}\right)\la(dv_1)\dotso\la(dv_n)m(dx_1)\dots m(dx_n),
		\end{align*}
		where $F\colon\K_0(\X)\to\R$ such that the above expression exists and $0$ denotes the zero measure. Then, $\CL$ is locally finite.
	\end{Example}

	\begin{Proposition}\label{cm_dense}
		Let $\rho\in\CM_\mathrm{lf}(\K_0(\X))$. The space $B_\mathrm{cm}(\K_0(\X))$ is dense in $L^1(\rho)$ as well as $L^1(f\rho)$, where $f$ is the density function from Corollary \ref{cordense}.
	\end{Proposition}
	\begin{proof}
		By Lemma \ref{bsdense} and Corollary \ref{cordense}, it is enough to show that functions in the spaces $B_\mathrm{bs}(\K_0(\X))$ and $\widetilde{B}_\mathrm{bs}(\K_0(\X))$ can be approximated by functions in $B_\mathrm{cm}(\K_0(\X))$, with the convergence occurring with respect to $L^1(\rho)$ and $L^1(f\rho)$, respectively. Consider $G\in B_\mathrm{bs}(\K_0(\X))$. Define the sequence $\{G_n\}_{n=0}^\infty$ as:
		\begin{displaymath}
			G_n(\eta):=G(\eta)\cdot\prod_{x\in\tau(\eta)}\mathbbm{1}_{[\frac{1}{n},n]}(|v_x|),\ \eta\in\K_0(\X).
		\end{displaymath}
		Then, we have
		\begin{align*}
			|G_n(\eta)|&\leq C\mathbbm{1}_{\K_0(\Lambda)}(\eta)\mathbbm{1}_{\{|\tau(\eta)|\leq N\}}\prod_{x\in\tau(\eta)}\mathbbm{1}_{[\frac{1}{n},n]}(|v_x|),
            \end{align*}
		where $C,\Lambda$ and $N$ are defined as in Definition \ref{locsupp}. This shows that $G_n\in B_\mathrm{cm}(\K_0(\X))$. Next, we  show that $G_n$ approximates $G$ as $n \to \infty$. Specifically,
        \begin{align*}
			\|G_n-G\|_{L^1}&=\int_{\K_0(\X)}|G_n(\eta)-G(\eta)|\rho(d\eta)
			\\
			&=\int_{\K_0(\X)}|G(\eta)|\left|1-\prod_{x\in\tau(\eta)}\mathbbm{1}_{[\frac{1}{n},n]}(|v_x|) \right|\rho(d\eta)
			\\
			&=\int_{\K_0(\X)}|G(\eta)|\rho(d\eta)-\int_{\K_0(\X)}|G(\eta)|\prod_{x\in\tau(\eta)}\mathbbm{1}_{[\frac{1}{n},n]}|v_x|\rho(d\eta).
		\end{align*}
		Since $G\in L^1(\rho)$, by Lebesgue's theorem, it suffices to show that $G_n\to G$ pointwisely. Fix $\eta\in\K_0(\X)$. Since $\tau(\eta)$ is finite, there exists $n_0\in\N$ such that for all $n\geq n_0$, $|v_x|\in[\frac{1}{n},n]$ for every  $x\in\tau(\eta)$ . Therefore,
		\begin{displaymath}
			\prod_{x\in\tau(\eta)}\mathbbm{1}_{[\frac{1}{n},n]}(|v_x|)=1\ \forall n\geq n_0,
		\end{displaymath}
		which implies that $G_n(\eta)=G(\eta)$ for all $n\geq n_0$, i.e. $G_n\to G$ pointwisely. The above arguments establish that $G_n\to G$ in $L^1$, completing the proof.
		
		The proof for the density of  $B_\mathrm{cm}(\K_0(\X))$  in $L^1(f\rho)$ follows the same reasoning. Specifically, the estimate for $G_n$ as above for  $G\in\widetilde{B}_\mathrm{bs}(\K_0(\X))$ is given by
		$$
		|G_n(\eta)|\leq C\mathbbm{1}_{\K_0(\Lambda)}(\eta)\mathbbm{1}_{\{|\tau(\eta)|\leq N\}}\prod_{x\in\tau(\eta)}|v_x| 
		\mathbbm{1}_{[\frac{1}{n},n]}(|v_x|)
		$$
		$$
		\leq Cn^N\mathbbm{1}_{\K_0(\Lambda)}(\eta)\mathbbm{1}_{\{|\tau(\eta)|\leq N\}}\prod_{x\in\tau(\eta)}\mathbbm{1}_{[\frac{1}{n},n]}(|v_x|),
		$$
		which also implies $G_n\in B_\mathrm{cm}(\K_0(\X))$.
	\end{proof}
	
	%Similar to the case of $\Pi_0(\R^d_0\times \X)$, we now want to consider the extension of the $K$-transform to $L^1$-type spaces.
	
\subsection{Correlation Measures on $\K_0(\X)$}

We aim at establishing correlation measures for $\K_0(\X)$ that correspond to probability measures $\mu$ on $(\K(\X),\B(\K(\X)))$. Additionally, we aim to expand the $K$-transform to $L^1$-type spaces for appropriate classes of measures. We will follow a similar approach to that used for $\Pi_0(\mathbb{R}_0^d \times\X)$.

Given the unique structure of the spaces $\K(\X)$ and $\K_0(\X)$, we must consider the characteristics of measures concerning the marks.
	\begin{Definition} \label{pre-kernel}
    The pre-kernel $\mathcal{K}$ derived from the $K$-transform is defined as a map:
		\begin{align*}
			\mathcal{K}&\colon\B_b(\K_0(\X))\times\K(\X)\to[0,\infty),
			\\
		&(A,\eta)\mapsto\mathcal{K}(A,\eta):=(K\mathbbm{1}_A)(\eta).
		\end{align*}
	\end{Definition}
	Similar to the case of $\Pi_0(\mathbb{R}_0^d \times\X)$ it can be shown that $\CK$ is a pre-kernel and the same extension result holds:
	\begin{Lemma}
		The pre-kernel $\mathcal{K}$ can be uniquely extended to a kernel on $\B(\K_0(\X))\times\K(\X)$.
	\end{Lemma}
	\begin{proof}
		Similar to the proof of Lemma \ref{prekernel_ext_pi}
	\end{proof}
	\begin{Proposition} \label{pre-kernel}
		Let $G\colon\K_0(\X)\to\R$ be a measurable function with $G\geq 0$ or $G\in B_\mathrm{ls}(\K_0(\X))$. Then
		\begin{displaymath}
			\int_{\K_0(\X)}G(\xi)\mathcal{K}(d\xi,\eta)=\sum_{\xi\Subset\eta}G(\xi)=(KG)(\eta).
		\end{displaymath}
	\end{Proposition}
	\begin{proof}
		Similar to the proof of Proposition \ref{prop_2.52}
	\end{proof}
	Proposition \ref{pre-kernel} shows how $K$-transform relates to the pre-kernel $\CK$ for more general functions.
 
	Now we can use the kernel $\CK$  to construct measures on $(\K_0(\X),\B(\K_0(\X)))$ corresponding to the probability measures on $(\K(\X),\B(\K(\X)))$ .
	\begin{Definition}
		Let $\mu$ be a probability measure on the space $(\K(\X),\B(\K(\X)))$. The corresponding correlation measure is defined as a measure on $(\K_0(\X),\B(\K_0(\X)))$ by the relation:
		\begin{displaymath}
			\rho_\mu(A):=\int_{\K(\X)}\mathcal{K}(A,\eta)\mu(d\eta).
		\end{displaymath}
	\end{Definition}
 In the preceding section, we focused on the group of locally finite correlation measures $\rho_\mu$ on $(\K_0(\X)$. Although it is common for such correlation measures to be mark-locally finite rather than locally finite in practical applications, studying the class of locally finite measures remains meaningful.
	\begin{Proposition}\label{finite_moments}
		Let $\mu$ be a probability measure on $(\K(\X),\B(\K(\X)))$. Then the corresponding correlation measure $\rho_\mu$ is locally finite if and only if the following holds: For any $\Lambda\in\B_c(\X)$ and $N\in\N$,
		\begin{displaymath}
			\int_{\K(\X)}|\tau(\eta)\cap\Lambda|^N\mu(d\eta)<\infty.
		\end{displaymath}
		A measure $\mu$ with the above property is said to have finite local moments of all order. The space of all such measures is denoted by $\mathcal{M}^1_\mathrm{fm}(\K(\X))$.
	\end{Proposition}
	
	\begin{proof}[Proof of Proposition \ref{finite_moments}] The proof works analogously to the one of Proposition \ref{finite_moments_pi}.
	\end{proof}
	
	We introduce an example of measures on $\K_0(\X)$, which are not locally finite but at least mark-locally finite. We examine a certain type of measures mentioned above. For a measure $\rho$ on $\K_0(\X)$, set
	\begin{displaymath}
		\widetilde{\rho}(d\eta):=f(\eta)\rho(d\eta),
	\end{displaymath}
	where $f\colon\K_0(\X)\to (0,\infty)$ is the density function defined as:
	\begin{displaymath}
		f(\eta)=\prod_{x\in\tau(\eta)}\frac{1}{|v_x|}.
	\end{displaymath}
	\begin{Lemma}
		Let $\rho$ be a locally finite measure. Then $\widetilde{\rho}$ is mark-locally finite.
	\end{Lemma}
	\begin{proof}
		Let $A\in\B_\mathrm{cm}(\K_0(\X))$. Then $A\in\B_b(\K_0(\X))$. Furthermore, there exists $a>0$ such that for all $\eta\in A$, we have $|v_x|\geq a$ for all $x\in\tau(\eta)$. Then
		\begin{align*}
			\widetilde{\rho}(A)&=\int_{\K_0(\X)}\mathbbm{1}_A(\eta)\widetilde{\rho}(d\eta)=\int_{\K_0(\X)}\mathbbm{1}_A(\eta)f(\eta)\rho(d\eta)
			\\
			&\leq\int_{\K_0(\X)}\mathbbm{1}_A(\eta)\max\left(1,\frac{1}{a^N}\right)\rho(d\eta)=\max\left(1,\frac{1}{a^N}\right)\rho(A)<\infty.
		\end{align*}
	\end{proof}
	
	\subsection{Correlation Functions on $\K_0(\X)$}
    As noted earlier, correlation measures in applications are typically mark-locally finite but not locally finite. Our objective is to establish a density function specifically for this class of mark-locally finite correlation measures on \(\K_0(\X)\). Consequently, our analysis will concentrate on this particular class of measures.
	%As previously noted, correlation measures in applications are typically mark-locally finite but not locally finite. We focus on demonstrating a density function for this particular class of mark-locally finite correlation measures on $\K_0(\X)$. Therefore, our attention will be devoted to this class of measures.
    
	\begin{Definition}
		Let $\mu$ be a probability measure on $(\K(\X),\B(\K(\X)))$.
		\begin{enumerate}
			\item Let $\Lambda\subset\R^d_0\times \X$. For $\eta\in\K(\X)$ of the form $\eta=\sum_{x\in\tau(\eta)}v_x\delta_x$, define the projection with marks as:
			\begin{displaymath}
				p_\Lambda(\eta)=\sum_{\substack{x\in\tau(\eta)\\(v_x,x)\in\Lambda}}v_x\delta_x.
			\end{displaymath}
			The projection measure is defined as:
			\begin{displaymath}
				\mu^\Lambda:=\mu\circ p_\Lambda^{-1}.
			\end{displaymath}
			\item The measure $\mu$ is called mark-locally absolutely continuous with respect to the  measure $\mu_\la$ if for any $\Lambda\subset\R^d_0\times \X$ compact, the measure $\mu^\Lambda$ is absolutely continuous with respect to $\mu_\la^\Lambda$.
		\end{enumerate}
	\end{Definition}
	We will use a pullback correlation measure on the Plato space to examine mark-locally finite correlation measures on $\K_0(\X)$, and we will use pullback correlation measure on the Plato space.
	\begin{Proposition}\label{loc_finite_pi_k}
		A correlation measure $\rho_\mu$ on $\K_0(\X)$ corresponding to a probability measure $\mu$ on $\K(\X)$ is mark-locally finite if and only if the measure $\rho_{\mu_{\CR^{-1}}}$ on $\Pi_0(\R^d_0\times \X)$ is locally finite.
	\end{Proposition}
	\begin{proof}
		Let $\rho_\mu$ be the correlation measure of a measure $\mu$ on $\K(\X)$. Then for a set $A\in\B_\mathrm{cm}(\K_0)$,
		\begin{align*}
			\rho_\mu(A)&=\int_{\K(\X)}\CK(A,\eta)\mu(d\eta)=\int_{\K(\X)}(K_\K\mathbbm{1}_A)(\eta)\mu(d\eta)
			\\
			&=\int_{\K(\X)}\left[K_\Pi(\CR^{-1}\mathbbm{1}_A)\right](\CR^{-1}\eta)\mu(d\eta)=\int_{\K(\X)}\left[K_\Pi\mathbbm{1}_{\CR^{-1}A}\right](\CR^{-1}\eta)\mu(d\eta)
			\\
			&=\int_{\Pi(\R^d_0\times \X)}(K_\Pi\mathbbm{1}_{\CR^{-1}A})(\gamma)\mu_{\CR^{-1}}(d\gamma)
			\\
			&=\rho_{\mu_{\CR^{-1}}}(\CR^{-1}A),
		\end{align*}
		where $\mu_{\CR^{-1}}$ is the pullback measure of $\mu$ under $\CR$. Reversing the calculations yields the converse result.
	\end{proof}
    Mark-local absolute continuity allows us to compare measures on $\K(\X)$ and $\Pi(\R^d_0\times \X)$. The following Lemma will be used to compare measures on $\K(\X)$ and $\Pi(\R^d_0\times \X)$.
	\begin{Lemma}\label{loc_ac_pi_k}
		Let $\mu$ as a probability measure on $\K(\X)$ be mark-locally absolutely continuous concerning $\mu_\la$. Then $\mu_{\CR^{-1}}$ on $\Pi(\R^d_0\times \X)$ is locally absolutely continuous with respect to $\pi_\sigma$.
	\end{Lemma}
	\begin{Proposition}\label{density_from_k_pi}
 
		Suppose $\mu$ is a probability measure on $\K(\X)$ that satisfies the following conditions: it is mark-locally absolutely continuous concerning $\mu_\la$, and its associated correlation measure $\rho_\mu$ is mark-locally finite. Then, we can conclude that $\mu_{\CR^{-1}}\in\CM^1_\mathrm{fm}(\Pi(\R^d_0\times \X))$, and it is locally absolutely continuous with respect to $\pi_\sigma$. Additionally, we can observe that $\rho_{\mu_{\CR^{-1}}}$ is absolutely continuous with respect to $\CL_\sigma$ and its correlation function is well-defined.
	\end{Proposition}
	\begin{proof}
		The proof follows from Proposition \ref{loc_finite_pi_k}, Lemma \ref{loc_ac_pi_k} and Proposition \ref{density_on_pi}.
	\end{proof}
	The results obtained above imply the existence of a density function.
	\begin{Theorem}
		Assume the conditions of Proposition \ref{density_from_k_pi}. Then the correlation function of $\mu$ exists, i.e. a function $k_\mu\colon\K_0(\X)\to\R$ such that $k_\mu$ is the density function of $\rho_\mu$ with respect to the $\CL_\sigma$.
	\end{Theorem}
	\begin{proof}
		By the Proposition  \ref{density_from_k_pi}, we obtain the existence of a correlation function for $\rho_{\mu_{\CR^{-1}}}$ on $\Pi_0(\R^d_0\times \X)$ with respect to $\CL_\sigma$.  We obtain the desired result by transferring this function by the reflection mapping $\CR$.
	\end{proof}
    From a mathematical perspective, it is beneficial to represent a hierarchical structure associated with a function \( k \colon \K_0(\X) \to \R \). This representation allows the function on \(\K_0(\X)\) to be expressed as a sequence of functions on \((\R^d_0 \times \X)^n\). Consequently, in practical applications, an evolution equation defined on an infinite-dimensional space can be reformulated as a sequence of evolution equations on finite-dimensional spaces.
%	In terms of practical applications, depicting a hierarchical structure corresponding to a function $k\colon\K_0(\X)\to\R$ is advantageous. This structure enables us to substitute a function on $\K_0(\X)$ with a series of functions on $(\R^d_0\times \X)^n$. As a result, in practical applications, we can replace an evolution equation on an infinite-dimensional space with a series of equations on finite-dimensional spaces.

	\begin{Definition}\label{def_corrfn_k}
		Let $k\colon\K_0(\X)\to\R$. The hierarchical structure corresponding to $k$ is defined as the sequence of symmetric functions $\{k^{(n)}\}_{n=0}^\infty$, $k^{(n)}\colon(\R^d_0\times \X)^n\to\R$ by:
		\begin{equation*}
			k^{(n)}(v_1,x_1,\dotsc,v_n,x_n):=
			\begin{cases}
				k(\sum_{i=1}^n v_i\delta_{x_i}),  &\text{ if }\ \eta=\sum_{i=1}^n v_i\delta_{x_i}\in\K_0^{(n)}(\X),
				\\
				0, &\text{ otherwise}.
			\end{cases}
		\end{equation*}
The function $k_\mu^{(n)}$ is referred to as the $n$-point correlation function of $\mu$. For convenience, we also write:
		\begin{displaymath}
			k^{(n)}(v_1,\dotsc,x_n):=k^{(n)}(v_1,x_1,\dotsc,v_n,x_n).
		\end{displaymath}
	\end{Definition}
We  define the following notion of a correlation function:
	\begin{Definition}
		The $n$-point correlation function on $\K_0(\X)$ with respect to positions is defined as:
		\begin{displaymath}
			\varkappa_{\mu,h}^{(n)}(x_1,\dotsc,x_n):=\int_{(\R^d_0)^n}<h,v_1>\dotsm <h,v_n> k^{(n)}_\mu(v_1,\dotsc,x_n)\la(dv_1)\dotso\la(dv_n),
		\end{displaymath}
		where $k^{(n)}$ is the $n$-point correlation function introduced in Definition \ref{def_corrfn_k}
		and $h\in \X$.
	\end{Definition}
	\begin{Remark}
		The function $\varkappa^{(n)}$ can be obtained by integration of $k^{(n)}$; therefore, we can proceed by only analyzing the latter.
	\end{Remark}
\section{Discussion}

We developed a rigorous mathematical framework to analyze classical continuous systems with singular velocity distributions, addressing a significant gap in modelling such systems. These distributions, represented by Radon measures with infinite mass, introduce challenges beyond standard approaches in configuration spaces. In particular, we focused on constructing and analyzing the phase space of configurations and measures, with specific emphasis on the Plato space \(\Pi(\mathbb{R}_0^d \times \mathbb{R}^d)\) and the cone of vector-valued discrete Radon measures \(\K(\mathbb{R}^d)\), connected via the reflection mapping $\CR$. This relationship is the foundation for developing harmonic analysis and studying associated measures and dynamics. Let us underline that the conceptualization of the Plato space introduces an idealized setting where configurations are locally finite subsets of \(\mathbb{R}_0^d \times \mathbb{R}^d\), which are required to satisfy the condition that, for any compact subset \(\Lambda \subset \mathbb{R}^d\), the total velocity of particles within \(\Lambda\) remains finite, expressed as:
\[
V_\Lambda(\gamma) = \sum_{x \in \tau(\gamma) \cap \Lambda} |v_x| < \infty, \quad \forall \gamma \in \Pi(\mathbb{R}_0^d \times \mathbb{R}^d),
\]
\(\tau(\gamma)\) being the projection of the configuration \(\gamma\) onto spatial positions. This constraint ensures the system avoids unbounded velocities within finite spatial regions, aligning with physical plausibility.
Moreover, since the reflection mapping $\CR$ connects the Plato space \(\Pi(\mathbb{R}_0^d \times \mathbb{R}^d)\) with the cone \(\K(\mathbb{R}^d)\) by mapping configurations \(\gamma \in \Pi(\mathbb{R}_0^d \times \mathbb{R}^d)\) to vector-valued Radon measures:
\[
\CR\gamma = \sum_{x \in \tau(\gamma)} v_x \delta_x, \quad \text{where } v_x \in \mathbb{R}_0^d,
\]
then it defines \(\K(\mathbb{R}^d)\) as the image of \(\Pi(\mathbb{R}_0^d \times \mathbb{R}^d)\) under $\CR$. The latter gives the possibility of forming the space of observable objects derived from the ideal configurations. Consequently, the topology of \(\K(\mathbb{R}^d)\) inherits its structure from the vague topology on the configuration space.
In the latter setting, a cornerstone of the analysis lies in the use of Poisson measures on \(\Pi(\mathbb{R}_0^d \times \mathbb{R}^d)\), defined via their Laplace transform:
\[
\int_{\Pi(\mathbb{R}_0^d \times \mathbb{R}^d)} \exp\left( \sum_{(v, x) \in \gamma} \psi(v, x) \right) \pi_{\sigma}(d\gamma)
= \exp \int_{\mathbb{R}_0^d \times \mathbb{R}^d} \left( e^{\psi(v, x)} - 1 \right) \sigma(dv, dx),
\]
where \(\sigma = \lambda \otimes m\) is an intensity measure combining a Radon measure \(\lambda(dv)\) on \(\mathbb{R}_0^d\) and a Lebesgue measure \(m(dx)\) on \(\mathbb{R}^d\), hence encapsulating the statistical properties of the system and forms the basis for deriving probability measures on \(\K(\mathbb{R}^d)\).
We also implemented the associated harmonic analysis via the K-transform to facilitate the transition between finite configurations in \(\Pi_0(\mathbb{R}_0^d \times \mathbb{R}^d)\) and infinite configurations in \(\Pi(\mathbb{R}_0^d \times \mathbb{R}^d)\). 

Indeed, it allows the extension of functions and measures from finite to infinite configurations, providing a powerful tool for analyzing the system's dynamics.
Moreover, introducing the measures on \(\K(\mathbb{R}^d)\) via the pushforward of measures on \(\Pi(\mathbb{R}_0^d \times \mathbb{R}^d)\) under $\CR$, allows us to study complex systems with singular velocity distributions. Accordingly, we derive a central result of the paper, i.e.: we proved that the image \(\sigma\)-algebra of \(\Pi(\mathbb{R}_0^d \times \mathbb{R}^d)\) under $\CR$ coincides with the \(\sigma\)-algebra of \(\K(\mathbb{R}^d)\), ensuring the measurability of the reflection mapping and the compatibility of the mathematical structures.
The results discussed above notably extend the scope of prior research on configuration spaces and Radon measures, especially those concentrating on positive measures in Riemannian contexts. Classical studies, such as those by Kondratiev et al., have explored the harmonic analysis of positive discrete Radon measures on Riemannian manifolds, representing particle configurations with purely spatial characteristics. These configurations are typically modelled as locally finite sets with well-behaved measures, facilitating the development of analytical techniques such as correlation functions and generating functionals. In contrast, the framework developed here addresses a broader and more complex class of systems: vector-valued Radon measures with singular distributions. These singularities arise when measures have infinite mass or exhibit significant anisotropy in velocity space. Specifically, velocity distributions of the form \(\lambda(dv) = |v|^{-\alpha} e^{-|v|^\beta} dv\), with \(\alpha \in [d, d+1)\) and \(\beta > 0\), present deep mathematical and physical consequences. These distributions enable the inclusion of systems where velocities can become arbitrarily large near specific regions of phase space, a scenario that cannot be described using traditional finite measures.
%The above-recalled results significantly expand the scope of earlier studies on configuration spaces and Radon measures, particularly those focusing on positive measures in Riemannian settings. Classical works, such as those by Kondratiev and colleagues, investigated the harmonic analysis of positive discrete Radon measures on Riemannian manifolds, where the measures describe configurations of particles with purely spatial properties. These configurations are typically modelled as locally finite sets with well-behaved measures, allowing for the development of robust analytical techniques such as correlation functions and generating functionals. In contrast, the framework developed here addresses a more general and challenging class of systems: vector-valued Radon measures with singular distributions. The singularity arises from introducing measures that possess infinite mass or exhibit significant anisotropy in the velocity space. Specifically, velocity distributions parameterized by \(\lambda(dv) = |v|^{-\alpha} e^{-|v|^\beta} dv\), where \(\alpha \in [d, d+1)\) and \(\beta > 0\), introduce profound mathematical and physical implications. Such distributions allow for the inclusion of systems where velocities can become arbitrarily large near certain regions of phase space, a scenario that cannot be captured using traditional finite measures.
In particular, the parameter \(\alpha\) determines the degree of singularity near \(v = 0\). For \(\alpha = d\), the measure \(\lambda(dv)\) transitions to a marginally finite scenario. In contrast, values of \(\alpha\) approaching \(d+1\) emphasize the singular nature of the distribution, leading to configurations that challenge standard assumptions of local finiteness. On the other hand, the parameter \(\beta\) controls the distribution's decay rate at large \(|v|\), with larger values corresponding to faster decay. The interplay between \(\alpha\) and \(\beta\) broadens the modelling flexibility, accommodating a wide variety of physical systems, e.g.:
when \(\beta = 2\), the measure approximates modified Maxwellian distributions relevant in statistical physics, while, for \(\beta < 2\), the distributions model systems with heavy-tailed velocities characteristic of certain astrophysical or turbulent systems.

The shift from positive measures to vector-valued measures introduces additional complexity due to the need to account for the vectorial nature of the marks (velocities) attached to spatial configurations. This extension requires new mathematical tools to handle the coupled spatial and velocity dependencies. For instance, the harmonic analysis developed for positive measures relies on scalar-valued generating functions and correlation measures, whereas vector-valued measures necessitate functionals and transforms that respect the vector structure.
Let us also underline that the introduction of singular velocity distributions also enhances the physical realism of the model, enabling the study of systems where unbounded velocities naturally occur such, e.g.:
Astrophysical Systems, where stars and galaxies often exhibit velocity distributions with heavy tails due to gravitational interactions. The singular measures proposed here can model such systems without artificial truncations of velocity magnitudes; Kinetic Theory: namely in the Boltzmann framework, velocity distributions often develop singular features due to collision dynamics, especially in dilute gases; Turbulence and Fluids: where singular velocity distributions appear in the statistical description of turbulent flows, where small-scale dynamics generate extreme velocities; etc.

The generalizations discussed above introduce additional complexity in the definition and analysis of measures on configuration spaces. Specifically, the transition from scalar-valued to vector-valued measures on \(\K(\mathbb{R}^d)\) requires a careful consideration of both topological and algebraic structures. Key to this is ensuring the measurability of the reflection mapping \(\CR: \Pi(\mathbb{R}_0^d \times \mathbb{R}^d) \to \K(\mathbb{R}^d)\) and the compatibility of the corresponding \(\sigma\)-algebras. The K-transform, originally developed for scalar measures, is extended here to accommodate vector-valued measures, providing robust harmonic analysis and probabilistic modeling. This extension enables the handling of singular and anisotropic velocity distributions, thus broadening the framework's applicability to various physical and mathematical problems. This includes exploring systems where classical assumptions, such as bounded velocities or locally finite configurations, do not hold, opening new research directions in studying complex systems. 

Future work will aim to extend this framework to interacting particle systems, where pairwise or higher-order interactions introduce additional complexities. Additionally, generalizing the framework to non-Euclidean geometries, such as Riemannian or hyperbolic spaces, will shed light on systems with curvature-dependent dynamics. Numerical implementations of the K-transform and associated measures could facilitate the application of these methods to real-world systems. Lastly, investigating stochastic evolution on \(\K(\mathbb{R}^d)\) in the context of (stochastic) mean-field games is a promising direction, linking the theoretical framework to dynamic processes.

\section*{Acknowledgments}
This work would not have been possible without the constant help, mathematical support and brilliant intuitions of our friend Yuri Kondratiev.
We shared the fundamental steps of this work with him, but unfortunately, we were not able to finish it all together: Prof. Kondratiev left us prematurely. His absence is undoubtedly strong, but not so much as to obscure his human greatness, even before his excellent scientific abilities.
Sit tibi terra levis, our dearest friend.


\begin{thebibliography}{999}
% Reference 1
\bibitem{AKR}
Albeverio, S.; Kondratiev, Y.;  Rockner, M.  Analysis and geometry on configuration spaces. {\em J. Funct. Anal. 154 (2), } {\bf 1998},  444--500.

\bibitem{B}
Bogoliubov, N.N. \textit{Problems of a Dynamical Theory in Statistical Physics.} {\em Gostekhisdat},  Moscow, {\bf 1946}. 

% Reference 3
\bibitem{CKL}
Conache, D.; Kondratiev, Y.;  Lytvynov, E.  Equilibrium diffusion on the cone of discrete Radon measures. {\em Potential Anal. 44 ,} {\bf 2016},  71--90.
% Reference 4
\bibitem{PhD}
Finkelshtein, D.; Kondratiev, Y.;  Kuchling, P.; Lytvynov, E.; Oliveira, M.J. Analysis on the cone of discrete Radon measures.
\url{https://doi.org/10.48550/arXiv.2312.03537}.
% Reference 5
\bibitem{KSS}
Da Silva, J.L.; Kondratiev, Y.;  Streit, L. Differential geometry on compound Poisson space. {\em Methods Funct. Anal. Topology 4(1),} {\bf 1998}, 32--58.
% Reference 6
\bibitem{Ki}
 Kingman, J.F.C. \textit{Poisson Processes}. 1st Edition, Oxford Studies in Probability , Clarendon Press, UK, {\textbf{1993}}.
% Reference 7
\bibitem{KLV}
Kondratiev, Y.; Lytvynov, E.;  Vershik, A. Laplace operators on the cone of Radon measures . {\em J.Funct. Anal. 269 (9),} {\bf 2015}, 2947--2976.
% Reference 8
\bibitem{MR1914839}
Kondratiev, Y.; Kuna,  T.  Harmonic analysis on configuration space. I. General theory. {\em Infin. Dimens. Anal. Quantum Probab. Relat. Top. 5(2),} {\bf 2002}, 201--233. 
\bibitem{MR2253724}
 Kondratiev, Y.;  Kuna,  T.; Oliveira, M. J. Holomorphic Bogoliubov functionals for interacting particle systems in continuum. {\em J. Funct. Anal. 238,} {\bf 2006}, 375--404. 
\bibitem{MR2426716} Kondratiev, Y.; Kutoviy, O;  Pirogov, S. Correlation functions and invariant measures in continuous contact model. {\em Infin. Dimens. Anal. Quantum Probab. Relat. Top. 11(2),} {\bf 2008}, 231--258. 

\bibitem{ph}
Kuchling, P.  Analysis and Dynamics on the Cone. \em{PhD thesis, Bielefeld University}, {\bf 2019}.

 \em
 
\bibitem {MR0323270} 
Lenard, A. 
Correlation functions and the uniqueness of the state in classical statistical mechanics.  \em{Comm. Math. Phys. 30}, {\bf 1973}, 35--44. 

\end{thebibliography}
\end{document}